\documentclass{lmcs}
\pdfoutput=1

\usepackage{lastpage}
\renewcommand{\dOi}{14(4:12)2018}
\lmcsheading{}{1--\pageref{LastPage}}{}{}%
{Feb.~22,~2017}{Nov.~01,~2018}{}

\usepackage{hyperref}
\usepackage{amssymb}
\usepackage{version}
\usepackage{xspace}
\usepackage{subcaption}
\usetikzlibrary{positioning}
\excludeversion{final}
\newcommand\csays[1]{\footnote{{\small\bf CB says:} #1}}

\newcommand\ksays[1]{\footnote{{\small\bf K says:} #1}}

\newcommand\reviewer[1]{} %\footnote{{\small\bf Reviewer 1:} #1}}
\newcommand\reviewerbis[1]{} %\footnote{{\small\bf Reviewer 2:} #1}}
\newcommand\response[1]{} %\footnote{{\small\bf Our response:} #1}}
\begin{final}
\renewcommand\csays[1]{}
\renewcommand\ksays[1]{}
\end{final}
\usepackage{nameref}
\makeatletter
\newcommand{\labelname}[1]{% \labelname{<stuff>}
  \def\@currentlabelname{#1}}%
\makeatother

\usepackage{mathpartir}

\newtheorem{infrule}{Rule}{\bfseries}{\itshape}

\newcommand{\infrulenew}[3][]{
$\inferrule[#1]{#2}{#3}$
}

\DeclareMathOperator{\opplus}{\texttt{+}\,}
\DeclareMathOperator{\opminus}{\texttt{-}\,}
\DeclareMathOperator{\opless}{\texttt{<}\,}
\DeclareMathOperator{\opgeq}{\texttt{>=}\,}

\DeclareMathOperator{\leavesfn}{Leaves}
\newcommand{\leaves}[1]{\leavesfn{(#1)}}
\DeclareMathOperator{\add}{add}
\DeclareMathOperator{\merge}{merge}

\DeclareMathOperator{\memintermfn}{Elements}
\newcommand{\Elements}[1]{\memintermfn(#1)}

\DeclareMathOperator{\varsoffn}{Vars}
\newcommand{\varsof}[1]{\varsoffn{#1}}
\DeclareMathOperator{\termsoffn}{Terms}

\newcommand{\setbuilder}[2]{\left\{#1\ \middle|\ #2\right\}}
\newcommand{\inparen}[1]{\left(#1\right)}
\newcommand{\inbrace}[1]{\left\{#1\right\}}

\newif \ifFANCYOPER
\FANCYOPERfalse

\newcommand{\cvc}{\textsc{cvc4}\xspace}

\newcommand{\define}[1]{\textsl{#1}}
\newcommand{\var}[1]{\mathcal{V}( #1 )}
\newcommand{\Mo}{\mathbf{I}}
\newcommand{\ent}[1][]{\models_{#1}}
\newcommand{\ff}{\bot}
\newcommand{\str}[1]{\mathcal{#1}}
\newcommand{\I}{\str{I}}

\newcommand{\unsat}{\ensuremath{\mathsf{unsat}}\xspace}

\newcommand{\sortFont}[1]{\mathsf{#1}}

\newcommand{\cardSort}{\sortFont{Card}}
\newcommand{\elemSort}{\sortFont{Element}}
\newcommand{\setSort}{\sortFont{Set}}

\newcommand{\theory}[1]{\mathfrak{#1}}

\newcommand{\ariththeory}{{\theory{T}_{A}}}

\newcommand{\settheory}{{\theory{T}_{S}}}

\ifFANCYOPER
\newcommand{\opemptyset}{\text{`}\emptyset\text{'}}
\else
\newcommand{\opemptyset}{\emptyset}
\fi
\newcommand{\opunion}{\sqcup}
\newcommand{\opinter}{\sqcap}
\newcommand{\opsetminus}{\setminus}
\newcommand{\opin}{\mathrel{\ooalign{$\sqsubset$\cr{$-$}}}}
\newcommand{\opnotin}{\not\opin}
\newcommand{\opequal}{\approx}

\newcommand{\opsubset}{\sqsubseteq}
\ifFANCYOPER
\newcommand{\opsingleton}[1]{\left\lfloor#1\right\rfloor} %for drafting
\else
\newcommand{\opsingleton}[1]{\left\{#1\right\}} %for drafting
\fi

\newcommand{\opcar}{\mathsf{card}}
\newcommand{\opcard}[1]{\opcar(#1)}

\newcommand{\card}[1]{\left|#1\right|}
\newcommand{\singleton}[1]{\inbrace{#1}}

\newcommand{\setconst}{\mathcal{S}}
\newcommand{\memconst}{\mathcal{M}}
\newcommand{\arithconst}{\mathcal{A}}
\newcommand{\graphconst}{\mathcal{G}}
\newcommand{\grapharithconst}{\hat{\graphconst}}

\newcommand{\setconstclosed}{\setconst^*}
\newcommand{\memconstclosed}{\memconst^*}

\newcommand{\termsof}[1]{\mathcal{T}}
\newcommand{\termsofsort}[2]{\termsoffn_{#1}(#2)}
\newcommand{\realtermsof}[1]{\termsoffn(#1)}

\newcommand{\closedaddfn}{\triangleleft}
\newcommand{\closedadd}[2]{#1 \triangleleft (#2)}

\newcommand{\interpretation}{\mathfrak{I}}

\newcommand{\rn}[1]{\textsc{#1}}

\newcommand{\ruleInterDownI}{\rn{Inter Down I}}
\newcommand{\ruleInterDownII}{\rn{Inter Down II}}
\newcommand{\ruleInterUpI}{\rn{Inter Up I}}
\newcommand{\ruleInterUpII}{\rn{Inter  Up II}}
\newcommand{\ruleInterUpSplit}{\rn{Inter split}}

\newcommand{\ruleUnionDownI}{\rn{Union Down I}}
\newcommand{\ruleUnionDownII}{\rn{Union Down II}}
\newcommand{\ruleUnionUpI}{\rn{Union Up I}}
\newcommand{\ruleUnionUpII}{\rn{Union  Up II}}
\newcommand{\ruleUnionUpSplit}{\rn{Union split}}

\newcommand{\ruleSetDifferenceDownI}{\rn{Set difference Down 1}}
\newcommand{\ruleSetDifferenceDownII}{\rn{Set difference Down 2}}
\newcommand{\ruleSetDifferenceDownIII}{\rn{Set difference Down 3}}
\newcommand{\ruleSetDifferenceUpI}{\rn{Set difference Up 1}}
\newcommand{\ruleSetDifferenceUpII}{\rn{Set difference Up 2}}
\newcommand{\ruleSetDifferenceUpIII}{\rn{Set difference Up 3}}
\newcommand{\ruleSetDifferenceSplit}{\rn{Set difference split}}

\newcommand{\ruleSingle}{{\rn{Singleton}}}
\newcommand{\ruleSingleMem}{\rn{Single Member}}
\newcommand{\ruleSingleNonMem}{\rn{Single Non-member}}
\newcommand{\ruleSetDiseq}{\rn{Set Disequality}}
\newcommand{\ruleEqUnsat}{\rn{Eq Unsat}}
\newcommand{\ruleSetUnsat}{\rn{Set Unsat}}
\newcommand{\ruleEmptyUnsat}{\rn{Empty Unsat}}
\newcommand{\ruleIntroEqRight}{\rn{Introduce Eq Right}}
\newcommand{\ruleIntroEqLeft}{\rn{Introduce Eq Left}}
\newcommand{\ruleIntroEqUnion}{\rn{Introduce Union}}
\newcommand{\ruleIntroEqInter}{\rn{Introduce Inter}}
\newcommand{\ruleIntroEqSetDiff}{\rn{Introduce Set difference}}
\newcommand{\ruleIntroEqCard}{\rn{Introduce Card}}
\newcommand{\ruleIntroEqSingle}{\rn{Introduce Singleton}}
\newcommand{\ruleIntroEqEmpty}{\rn{Introduce Empty Set}}
\newcommand{\ruleMergeEqI}{\rn{Merge Equality I}}
\newcommand{\ruleMergeEqII}{\rn{Merge Equality II}}
\newcommand{\ruleMergeEqIII}{\rn{Merge Equality III}}
\newcommand{\ruleAriContra}{\rn{Arithmetic contradiction}}
\newcommand{\ruleGuessEmpty}{\rn{Guess Empty Set}}
\newcommand{\ruleMemArrange}{\rn{Members Arrangement}}
\newcommand{\ruleGuessLB}{\rn{Guess Lower Bound}}
\newcommand{\rulePropMinsize}{\rn{Propagate Minsize}}

\newcommand{\ruleIntroSingleton}{\rn{Introduce Singleton}}
\newcommand{\ruleIntroEmtpy}{\rn{Introduce Empty Set}}

\newcommand{\derivtree}{\mathcal{D}}
\newcommand{\derivation}{\mathbf{D}}

\newcommand{\aset}{s}
\newcommand{\bset}{t}
\newcommand{\cset}{u}
\newcommand{\dset}{v}
\newcommand{\asetv}{S}
\newcommand{\bsetv}{T}
\newcommand{\csetv}{U}
\newcommand{\cardaset}{c_s}
\newcommand{\cardbset}{c_t}

\newcommand{\vertices}{V(\graphconst)}
\newcommand{\edges}{E(\graphconst)}
\newcommand{\neigh}[1]{C(#1)}

\newcommand{\allsetterms}{\mathcal{T}}
\newcommand{\allleaves}{\leaves{\graphconst}}
\newcommand{\structure}{\mathfrak{I}}
\newcommand{\structureS}{\mathfrak{S}}
\newcommand{\structureA}{\mathfrak{A}}
\newcommand{\nonemptyleavesfn}{\mathcal{L}}
\newcommand{\nonemptyleaves}[1]{\nonemptyleavesfn(#1)}

\newcommand{\eqls}{\mathcal{E}}
\newcommand{\equivcl}[2]{\left[#2\right]_{#1}}
\newcommand{\triviallyimply}{\Rrightarrow}
\newcommand{\nottriviallyimply}{\not\Rrightarrow}

\newcommand{\refgraphprop}[1]{Proposition \ref{prop:graph}, property \ref{prop:graph:prop#1}}

\newcommand{\termreln}{\succ}

\newtheorem{restr}[thm]{Restriction}
\theoremstyle{remark}
\newtheorem{case}{Case}

\begin{document}

\title
%    [Decision Procedure for Finite Sets and Cardinality] %short
%    {A New Decision Procedure for Finite Sets and
%      Cardinality Constraints in SMT}
[Reasoning with finite sets and cardinality]
{Reasoning with finite sets and \\
cardinality constraints in SMT}

\author[Bansal et al.]{Kshitij Bansal\rsuper{a}}	%required
% (NYU &) ? Do I need to metion prvs affiliation as well when work was done?
\address{\lsuper{a}Google, Inc.}	%required
\email{kbk@google.com}  %optional
%\thanks{thanks 1, optional.}	%optional
\thanks{This work was partially supported by NSF grants 1228765,
  1228768, and 1320583. The first author was at New York University
  when this work was completed.}
% Thanks can be split-up author-wise if needed. This is what the doc says:
%% \section{Multiple authors}
%%   In papers with multiple authors several points need to be mentioned.
%%   Do not worry about footnote signs that will link author $n$ to
%%   address $n$ and the optional thanks $n$.  This will be taken care of
%%   by the layout editor.  Even if authors share an affiliation and part
%%   of an email address, they should follow the strict scheme outlined
%%   above and list their data individually.  The layout editor will
%%   notice duplication of data and can then arrange for more
%%   space-efficient formatting.  Alternatively, Authors can write ``same
%%   data as Author n'' into some field to alert the layout editor.
%%   Unfortunately, so far we have not been able to devise a system that
%%   automatically takes care of these issues.  But once the layout
%%   editor is made aware of some duplication, he can take some fairly
%%   simple measures to adjust the format accordingly.  Placing the
%%   responsibility on the layout editor insures that these formatting
%%   issues are handled uniformly in different papers and that the
%%   authors do not have to second-guess the Journal's policy.

\author[]{Clark Barrett\rsuper{b}}	%optional
%\address{address2; addresses should be duplicated when authors share an affiliation}	%optional
\address{\lsuper{b}Department of Computer Science, Stanford University}
%\email{author2@email2; ditto for email addresses}  %optional
\email{barrett@cs.stanford.edu}
%\thanks{Same data as author 1.}	%optional

\author[]{Andrew Reynolds\rsuper{c}}	%optional
\address{\lsuper{c}Department of Computer Science, The University of Iowa}	%optional
\email{andrew-reynolds@uiowa.edu}
%\thanks{thanks 3, optional.}	%optional
%\thanks{Same data as author 1.}	%optional

\author[]{Cesare Tinelli\rsuper{c}}	%optional
\address{\vskip-7pt}
\email{cesare-tinelli@uiowa.edu}
%\thanks{thanks 3, optional.}	%optional
%\thanks{Same data as author 1.}	%optional

%% etc.

%% required for running head on odd and even pages, use suitable
%% abbreviations in case of long titles and many authors:

%% mandatory lists of keywords and classifications:
\keywords{Satisfiability modulo theories, Finite sets, Decision procedures}
\subjclass{Theory of computation: Automated reasoning}
%% \titlecomment{OPTIONAL comment concerning the title, \eg, if a variant
%% or an extended abstract of the paper has appeared elsewehere}
\titlecomment{This article extends~\cite{BRBT16} with
additional calculus rules and a full proof of correctness.}
%%%%%%%%%%%%%%%%%%%%%%%%%%%%%%%%%%%%%%%%%%%%%%%%%%%%%%%%%%%%%%%%%%%%%%%%%%%

%% the abstract has to PRECEED the command \maketitle:
%% be sure not to issue the \maketitle command twice!

%% \begin{abstract}
%%   \noindent The abstract has to preceed the maketitle command.  Be
%%   sure not to issue the maketitle command twice!  Preferably, the
%%   abstact should consist of plain ASCII text, without mathematical
%%   expressions or \TeX-commands, including explicit references using
%%   the cite command.  Presently we are not able to automatically
%%   extract an abstract containing such data and relyably turn it into
%%   html code.  If you cannot meet these criteria, it is your
%%   responsibility to provide us with an html-version of your abstract.
%%   Please keep the abstract fairy short to prevent it from spilling
%%   over to the second page!
%% \end{abstract}
\begin{abstract}
We consider the problem of deciding the satisfiability of quantifier-free formulas
in the theory of finite sets with cardinality constraints.
Sets are a common high-level data structure used in
programming; thus, such a theory is useful for modeling program
constructs directly. More importantly, sets are a basic construct of
mathematics and thus natural to use when formalizing the
properties of computational systems.
We develop a calculus describing a modular combination of a procedure for
reasoning about membership constraints with a procedure for reasoning about
cardinality constraints. Cardinality reasoning involves tracking how
different sets overlap. For efficiency, we avoid considering Venn regions
directly, as done in previous work.  Instead, we develop a
novel technique wherein potentially overlapping regions are considered
incrementally as needed, using a graph to track the interaction 
among the different regions.
The calculus has been designed to facilitate its implementation
within SMT solvers based on the DPLL($T$) architecture.
Our experimental results demonstrate that the new techniques are
competitive with previous techniques and can scale much better 
on certain classes of problems.
\end{abstract}

\maketitle

%!TEX root =  pap.tex

\section{Introduction}

%\ctsaysi{To be revised up to Formal Preliminaries}

Satisfiability modulo theories (SMT) solvers are at the heart of many
formal methods tools. One of the reasons for their popularity is that
fast, dedicated decision procedures for fragments of
first-order logic that SMT solvers implement are extremely useful for
reasoning about constructs common in hardware and software
verification.  In particular, they provide a good balance between
speed and expressiveness. Common fragments include theories such as
bitvectors, arithmetic, and arrays, which are useful for modeling
basic constructs as well as for performing general reasoning.

As the use of SMT solvers has spread, there has been a corresponding demand for
SMT solvers to support additional useful theories.  Although it is possible to
encode finitely axiomatizable theories using quantifiers, the performance and
robustness gap between a custom decision procedure and an encoding using
quantifiers can be quite significant.

In this paper, we present a new decision procedure for a fragment of finite set theory. 
Our main motivation is that sets are a
common abstraction used in programming.
As with other general-purpose SMT theories such as the theories of arrays and bitvectors, 
the theory of finite sets is useful for modeling a variety of program constructs.
Sets are also used directly in high-level programming languages such as SETL~\cite{Schwartzetal86} and in
specification languages such as Alloy~\cite{jackson2012software}, B~\cite{abrial2005b} and Z~\cite{AbrialSM80}.
More generally, sets are a basic construct in mathematics
and come up quite naturally when trying to express properties of systems.

While the full language of set theory is undecidable, many interesting
fragments are known to be decidable.  We present a calculus for the theory of finite sets 
which can handle basic set operations, such as membership, union, intersection, and
difference, and which can also reason efficiently about set cardinalities
and linear constraints involving them.  
The calculus is explicitly designed for easy integration into the DPLL($T$) framework~\cite{NieOT-JACM-06}.
We briefly describe our implementation in the DPLL($T$)-based SMT solver \cvc
and an initial experimental evaluation of this implementation.
 
%!TEX root =  pap.tex

\subsection{Related work}

%% Even though a more comprehensive historial background is given Section
%% \ref{sec:sets:bibnote}, 
%% In this section, we compare the calculus with
%% work which is most closely related to the fragment we consider.
In the SMT community, the desire to support a theory of finite sets with
cardinality goes at least as far back as a proposal by Kr\"oning et al.~\cite{kroning2009proposal}.
That article focuses on formalizing the semantics and representation of
the theory within the context of the SMT-LIB standard, rather than on a
decision procedure for deciding it.

%The full language of set theory is undecidable, but there
There is a stream of research on exploring decidable fragments of set theory 
(often referred to in the
literature as syllogistics)~\cite{cantone01}.
One such subfragment is MLSS, more precisely, the ground set-theoretic
fragment with basic Boolean set operators (union, intersection, set
difference), singleton operator and membership predicate.
A tableau-based procedure for this fragment was introduced
by Cantone and Zarba~\cite{cantone1998}.
The part of our calculus covering this fragment builds on their work.
De Moura and Bj{\o}rner presented an extension of the theory of 
arrays~\cite{de2009generalized}
that can be used to encode the MLSS fragment.  
However, this approach cannot be used to encode cardinality constraints.

In this paper, we consider an extension of the MLSS fragment with set cardinality operations,
whose decidability was established by Zarba~\cite{Zarba02frocos,Zarba2005}.
The decision procedure described by Zarba involves making an upfront guess 
that is exponential
in the number of set variables, making it non-incremental and highly
impractical.  That said, the focus of that work is on establishing
decidability and not on providing an efficient procedure.

Another closely related logical fragment is the Boolean
Algebra and Presburger Arithmetic (BAPA) fragment, for which
several algorithms have been proposed \cite{KNR06,KR07,SSK11}. Though
BAPA does not have the membership predicate or the singleton operator
in its language, Suter et al.~\cite[Section 4]{SSK11} show how one can
generalize their algorithm for such reasoning. Intuitively, 
singleton sets can be simulated by imposing a cardinality constraint $\opcard{X}=1$.
Similarly, membership constraints of the form $x \opin S$ can be encoded 
as $X \opsubset S$ by introducing a singleton subset $X$.
This reduction can lead to significant inefficiencies, however.  Consider the following simple example:
$x \opin S_1 \opunion \inparen{S_2 \opunion  \inparen{\ldots \opunion \inparen{S_{99} \opunion S_{100}}}}$.
%$x \opnotin S_1$, $x \opnotin S_2$, $\ldots$, $x \opnotin S_{100}$.
%
In our
calculus, a straightforward repeated application of one of the rules for set unions can determine the satisfiability of this constraint.
In contrast, in a reduction to BAPA, membership reasoning is
reduced to reasoning about cardinalities of different sets.
For example, the algorithm in \cite{SSK11} will reduce the problem
to an arithmetic problem involving variables for $2^{101}$ Venn
regions derived from $S_1$, $S_2$, $\ldots$, $S_{100}$, and the
singleton set introduced for $x$.

The broader point is that reasoning about the cardinalities of Venn regions is the
main bottleneck for this fragment. As we show in our calculus, it is possible
to avoid using Venn regions for membership predicates by instead reasoning
about them directly.  For explicit cardinality constraints, our calculus
minimizes the number of Venn regions that need to be considered by reasoning
about only a limited number of relevant regions introduced lazily.

A procedure for cardinality constraints over multisets is considered in~\cite{DBLP:conf/vmcai/PiskacK08}.
A recent procedure for reasoning about sets and measure functions is given by Bender et al~\cite{DBLP:conf/cade/BenderS17},
which also relies on a reduction from set reasoning to arithmetic reasoning.
Reasoning about sets with cardinality constraints in the context 
of invariant checking for bounded model checking is considered by Alberti et al.~\cite{DBLP:conf/cade/AlbertiGP16},
and in the context of invariant synthesis by von Gleissenthall et el.~\cite{DBLP:conf/pldi/GleissenthallBR16}.
These works too rely on reductions to arithmetic and do not involve the use of dedicated decision procedures for 
sets in SMT solvers.
Other procedures for reasoning about sets include a
unification-based approach by Cristi{\'{a}} et al.~\cite{DBLP:conf/cav/CristiaR16}.

The theory we consider in this paper can be seen as the combination
of Presburger arithmetic with the theory of finite sets, with the cardinality
operator acting as a \emph{bridging function} between the two theories.
Decision procedures for non-disjoint combinations of theories 
with bridging functions have been studied 
by Sofronie-Stokkermans~\cite{Sofronie-Stokkermans2009} and 
Chocron et al.~\cite{DBLP:conf/cade/ChocronFR15}.
Their main contribution is the identification of restrictions on the theories
and the development of combination methods that allow one to construct
a decision procedure for the combined theory as a modular combination 
of the decision procedures for the component theories.
That work is mostly limited to cases of bridging functions from the theory
of algebraic datatypes to other theories, where the bridging function is definable
by recursion over constructor terms. 
It does not apply to our setting because neither our source theory 
nor the bridging function match those requirements.
Our approach is similar though in that it tries to separate as much as 
possible the reasoning about sets proper from the reasoning about their cardinality,
so as to leverage off-the-shelf linear integer arithmetic solvers
in order to reason about cardinalities. 

\subsection{Formal Preliminaries}
%------------------------------------------------------------------------------
%\ctsays{Possibly to be extended}
We work in the context of many-sorted first-order logic with equality.
We assume the reader is familiar with the following notions:
signature, term, literal, formula, free variable, 
interpretation, and satisfiability of a formula in an interpretation
(see, e.g.,~\cite{BarSST-09} for more details).
Let $\Sigma$ be a many-sorted signature.
We use $\opequal$ as the (infix) logical symbol for equality
for all sorts in $\Sigma$
and always interpret it as the identity relation.
%We write $s \opdisequal t$ as an abbreviation of $\lnot\,s \opequal t$.
If $e$ is a term or a formula, we denote by $\var e$ the set 
of $e$'s free variables, extending the notation to tuples and sets 
of terms or formulas as expected.

If $\varphi$ is a $\Sigma$-formula and $\I$ a $\Sigma$-interpretation,
we write $\I \models \varphi$ if $\I$ satisfies $\varphi$.
If $t$ is a term, we denote by $t^\I$ the value of $t$ in $\I$.
A \define{theory} is a pair $T = (\Sigma, \Mo)$, where 
$\Sigma$ is a signature and  $\Mo$ is a class of $\Sigma$-interpretations
that is closed under variable reassignment
(i.e., every $\Sigma$-interpretation that differs from one in $\Mo$
only in how it interprets the variables is also in $\Mo$).
We refer to $\Mo$ as the \define{models} of $T$.
A $\Sigma$-formula $\varphi$ is 
\define{satisfiable} (resp., \define{unsatisfiable}) \define{in $T$} 
if it is satisfied by some (resp., no) interpretation in $\Mo$.
A set $\Gamma$ of $\Sigma$-formulas \define{entails in $T$} a $\Sigma$-formula $\varphi$,
written $\Gamma \ent[T] \varphi$,
if every interpretation in $\Mo$ that satisfies all formulas in $\Gamma$
satisfies $\varphi$ as well.
We write $\ent[T] \varphi$ as an abbreviation for $\emptyset \ent[T] \varphi$.
We write $\Gamma \ent \varphi$ to denote that 
$\Gamma$ entails $\varphi$ in the class of all $\Sigma$-interpretations.
The set $\Gamma$ is \define{satisfiable in $T$} if $\Gamma \not\ent[T] \ff$
where $\ff$ is the universally false atom.
Two $\Sigma$-formulas are \define{equisatisfiable in $T$}
if for every model $\str I$ of $T$ that satisfies one, there is a model 
of $T$ that satisfies the other and differs from $\str I$ at most 
over the free variables not shared by the two formulas.
When convenient, we will tacitly treat a finite set of formulas 
as the conjunction of its elements and vice versa. 

%%% Local Variables:
%%% mode: latex
%%% TeX-master: "pap.tex"
%%% End:

%!TEX root =  pap.tex

%==============================================================================
\section{A Theory of Finite Sets with Cardinality} \label{formalsetup}
%==============================================================================
\begin{figure}[t]
\flushleft
\textbf{Constant and function symbols:}
\medskip

\(
\begin{array}{l@{\hspace{2em}}l@{\hspace{1.5em}}l@{\hspace{1.5em}}l@{\hspace{1.5em}}l}
 & \multicolumn{2}{l}{n: \cardSort \text{\;\ for all } n \in \mathbb{N}} 
 & \opminus: \cardSort \to \cardSort 
 & \opplus: \cardSort \times \cardSort \to \cardSort
 \\[1ex]
 & \opemptyset: \setSort
 & \opcar : \setSort \to \cardSort
 & \opsingleton{\cdot} : \elemSort \to \setSort
 & \opunion, \opinter, \opsetminus : \setSort \times \setSort \to \setSort
\end{array}
\)
\bigskip

\textbf{Predicate symbols:}
\medskip

\(
\begin{array}{l@{\qquad}l@{\qquad}l@{\qquad}l@{\qquad}l}
 & \opless: \cardSort \times \cardSort 
 & \opgeq: \cardSort \times \cardSort 
 & \opsubset\, : \setSort \times \setSort
 & \opin\, : \elemSort \times \setSort
\end{array}
\)
\caption{The signature of $\settheory$.}
\label{fig:symbols}
\end{figure}

We are interested in a typed theory $\settheory$ of finite sets with cardinality.
In a more general logical setting,
this theory would be equipped with a parametric set type, with a type parameter 
for the set's elements, and a corresponding collection of polymorphic set
operations.\footnote{In fact, this is the setting supported in our implementation in \cvc.}
For simplicity here, we will describe instead a many-sorted theory of sets
of sort $\setSort$ whose elements are all of sort $\elemSort$.
The theory $\settheory$ can be combined with any other theory $\theory T$ 
in a standard way, i.e., Nelson-Oppen-style, 
by identifying the $\elemSort$ sort with a sort in $\theory T$ but
with the restriction that the sort must be interpreted in $\theory T$ as an infinite
set.\footnote{An extension that allows the sort to be interpreted as finite
  by relying on polite combination~\cite{JB10-LPAR} is left to
  future work.}  
Note that the many-sorted setting limits us to sets of elements of the same type
(so sets such as $\{1,\, \{2,3\},\, \{\{5\}\} \}$ are not representable).
Also, we limit our language to consider only \emph{flat} sets 
(i.e., no sets of sets of integers, say)
although this restriction can be lifted by combining
$\theory T$ with (copies of) itself using Nelson-Oppen combination.
More generally, an input having set constraints over multiple element types $T_1, \ldots, T_n$
can be handled by invoking $n$ copies of our procedure for these sorts
and combining them in the standard way.
The theory $\settheory$ has also a sort $\cardSort$ for terms denoting set cardinalities.
Since we consider only finite sets, all cardinalities will be natural numbers.
%
%\ksays{Beautify. Figure \ref{fig:symbols} items are not very well aligned vertically. Can be improved. Possibly by using align instead or array.}

Atomic formulas in $\settheory$ are built over a signature with these three sorts,
and an infinite set of variables for each sort.
Modulo isomorphism, $\settheory$ is the theory of a single many-sorted structure,
and its models differ in essence only on how they interpret the variables.
Each model of $\settheory$ interprets $\elemSort$ 
as some \emph{countably infinite} set $E$,
$\setSort$ as the set of \emph{finite} subsets of $E$, and
$\cardSort$ as $\mathbb{N}$.
The signature of $\settheory$ has the following predicate and function symbols,
summarized in Figure~\ref{fig:symbols}:
the usual symbols of \emph{linear} integer arithmetic,
the usual set composition operators, an empty set ($\opemptyset$) and 
a singleton set ($\opsingleton{\cdot}$) constructor,\footnote{%
We will use $\emptyset$, $\{$, and $\}$ also to denote sets at the meta level.
The difference between their two uses should be clear from context.
}
and a cardinality operator 
($\opcard{\cdot}$), all interpreted as expected.
The signature includes also symbols for the cardinality comparison ($\opless, \opgeq$),
subset ($\opsubset$) and membership ($\opin$) predicates.

We call \define{set term} any term of sort $\setSort$, and
\define{cardinality term} any term of sort $\cardSort$ with no occurrences 
of $\opcard{\cdot}$.
% Kshitij: We use "set term" not just for these definitions, but in the rest of the
% article as well. That is why it is problematic to include 
%  "or of the form $\opcard \aset$"
% in the definition of a set term.
%A set term is \define{atomic} if it is a variable or ....
A \define{set constraint} is an atomic formula of the form 
$\aset \opequal \bset$, $\aset \opsubset \bset$, $e \opin \aset$ or their negation,
with $\aset$ and $\bset$ set terms or of the form $\opcard \aset$, and $e$ a term of sort $\elemSort$.
A \define{cardinality constraint} is a [dis]equality $[\lnot] c \opequal d$ or 
an inequality $c \opless d$ or $c \opgeq d$ where $c$ and $d$ are cardinality terms.
An \define{element constraint} is a [dis]equality $[\lnot] x \opequal y$
where $x$ and $y$ are variables of sort $\elemSort$.
A \define{$\settheory$-constraint} is a set, cardinality or element constraint.
We write $u \not\opequal v$ and $e \opnotin \bset$ respectively 
as an abbreviation of $\lnot u \opequal v$ and $\lnot e \opin \bset$.

We use $x$, $y$ for variables of sort $\elemSort$; $\asetv$,
$\bsetv$, $\csetv$ for variables of sort $\setSort$; $\aset$, $\bset$,
$\cset$, $\dset$ for terms of sort $\setSort$; and $c$ with subscripts
for variables of sort $\cardSort$. Given $\mathcal{C}$, a set of
constraints, $\varsof(\mathcal{C})$
(respectively, $\realtermsof{\mathcal{C}}$) denotes the set of variables
(respectively, terms) in $\mathcal{C}$.
For notational convenience, we fix an injective mapping from terms of sort $\setSort$
to variables of sort $\cardSort$ that allows us to associate to each set term $s$ 
a unique cardinality variable $c_s$.

We are interested in checking the satisfiability in $\settheory$ of conjunctions 
of $\settheory$-constraints.
While this problem is decidable, it has high worst-case time complexity~\cite{Zarba02frocos}.
So our efforts are in the direction of producing a solver for $\settheory$-constraints
that is efficient in practice, in addition to being correct and terminating.
Our solver relies on the modular combination of a solver for set constraints and
an off-the-shelf solver for linear integer arithmetic, 
which handles arithmetic reasoning over set cardinalities.
%cb: This is not quite true: the arithmetic solver is used as a black box - it
%does not need to send any information to the set solver.  Also, it is sent
%more than just equalities between $\opcard S$ terms - also inequalities and linear terms for example.
%The combination between the two solvers is achieved, Nelson-Oppen style, 
%by exchanging equalities over shared terms, which however are not variables, 
%as in traditional combination procedures~\cite{TinHar-FROCOS-96,NelOpp-TPLS-79},
%but terms of the form $\opcard S$ where $S$ is a variable.

%%% Local Variables:
%%% mode: latex
%%% TeX-master: "pap.tex"
%%% End:

%!TEX root =  pap.tex
\section{A Calculus for the Theory}
\label{sec:sets:calculus}

In this section, we describe a tableaux-style calculus capturing
the essence of our combined solver for $\settheory$.
As we describe in the next section, that calculus admits a proof procedure
that decides the satisfiability of $\settheory$-constraints.

\begin{restr}\em 
\label{restr}
For simplicity, we consider as input to the calculus
only finite sets $\mathcal C$ of constraints whose set constraints are 
in \define{flat form}.
The latter are (well-sorted) set constraints of the form
$\asetv \opequal \bsetv$, $\asetv \not\opequal \bsetv$,
$\asetv \opequal \opemptyset$,
$\asetv \opequal \opsingleton{x}$,
$\asetv \opequal \bsetv \opunion \csetv$,
$\asetv \opequal \bsetv \opinter \csetv$,
$\asetv \opequal \bsetv \opsetminus \csetv$,
$x \opin \asetv$, $x \opnotin \asetv$, or
$c_\asetv \opequal \opcard{\asetv}$,
%or
%$x \opequal y$, $x \not\opequal y$
where $\asetv$, $\bsetv$, $\csetv$, $c_\asetv$, and $x$ are variables 
of the expected sort.
We also assume that any set variable $\asetv$ of $\mathcal C$ appears in at most 
one union, intersection or set difference term.
%and has and associated constraint of the form $c_\asetv \opequal \opcard{\asetv}$ in $E$.
Thanks to common equisatisfiability-preserving transformations
all of these assumptions can be made without loss of generality~\cite[Chapter 10]{cantone01}.
%CT added the following in response to a review's commment
These transformations include intermediate steps that replace constraints 
of the form $s \opsubset t$ with $s \opequal (s \opinter t)$.
They also include steps that 
replace each occurrence $i$ of the same term $t$ in union, intersection or set difference terms 
by a fresh variable $\bsetv_i$ 
while adding the equality constraint $\bsetv_i \opequal t$.

\end{restr}
\reviewerbis{
Restriction 3.1

First, a small development of the ``common satisfiability-preserving
transformations'' would be welcome. In particular for ``any set variable S of C
appears in at most one union, intersection or set difference term''. If S appears
in several of them, you give each occurrence different names and add the
equalities between the names?

On that note, how does that affect the number of vertices in the graph, if this
appears to be a bottleneck?

This is related to the importance of Restriction 3.1 on that graph, which I
hadn't fully measured on the first read. I had unwittingly reasoned with a
semantical understanding of the expressions in the graph vertices, or at least
considered intersections and unions to be AC operators, so that, when adding an
edge to $(T \opinter U)$ and later adding an edge to $(U\opinter T)$, the vertices would
be identified, thus limiting the number of vertices in the graph. If I
understand correctly, this can never happen because set variables only appear at
most once in such constructs, so the above would be written $(T \opinter U)$ and
$(U'\opinter T')$ with $T=T'$ and $U=U'$. So in effect, vertices would be duplicated if
they only differ by AC operations (right?), which I feared could be a source of
inefficiency if the number of vertices is critical.

But in any case, it seems that Restriction 3.1 is also what you use to make sure
that, once a vertex is no longer a leaf, no outgoing edge is ever going to be
added to that vertex in the future (otherwise, rules like INTRODUCE INTER and
INTRODUCE SET DIFFERENCE could well add new edges to non-leaves), and this is
why the property that a vertex is the disjoint union of its children is
maintained. This is quickly explained on page 9, although for me it became
useful to realise that "edges from a set variable node are added to the
graph only once" actually applies to every node. Isn't that right?
It became useful to me when I read the proof of termination (see below).

Related to this, I wondered why the rules INTRODUCE UNION, INTRODUCE INTER and
INTRODUCE SET DIFFERENCE feature the condition that the union / intersection /
set difference is not already a vertex (the rule INTRODUCE EMPTY SET does not,
because it is understood that the rule does not apply if the graph is
unchanged). I understand that an intersection can be added to the graph by rule
MERGE EQUALITY II, and allowing the later application of INTRODUCE INTER on that
intersection would change the graph (so you forbid it explicitly). But for set
difference or union? It seems to me that if it is already a vertex, then Fig. 5
is already included in the graph (in case of set difference, the
"intersection+difference" part of Fig.5). Is that incorrect? And is it also
correct that union nodes are never leaves?
}

The calculus is described as a set of derivation rules
which modify a \define{state} data structure. 
A state is either the special state \unsat or a tuple of the form
$\langle \setconst, \memconst, \arithconst, \graphconst \rangle$,
where 
\begin{itemize}
\item $\setconst$ is a set of set constraints,
\item $\memconst$ is a set of element constraints,
\item $\arithconst$ is a set of cardinality constraints, and 
\item $\graphconst$ is a directed graph over set terms with nodes $\vertices$
and edges $\edges$.
\end{itemize}
Initial states have the form 
$\langle \setconst_0, \memconst_0, \arithconst_0, \graphconst_0 \rangle$
where $\graphconst_0$ is the empty graph and 
$(\setconst_0, \memconst_0, \arithconst_0)$ is a partition of a given set of constraints
$\mathcal C$ satisfying Restriction~\ref{restr}.

Since cardinality constraints can be processed by a standard arithmetic solver,
and element constraints %can be processed 
by a simple equality solver,\footnote{
Recall that $\settheory$ has no terms of sort $\elemSort$ besides variables.
}
we present and discuss only rules that deal with set constraints.

%\ksaysi{@Cesare: R1 says meaining of 'guarded assignment form' isn't introduced, though it appears in the line following the term. Could you clarify this, or perhaps add a reference to a standard source?}
%
The derivation rules are provided in Figures~\ref{fig:set-rules1} % \ref{fig:set-rules2}, \ref{fig:introduce}, \ref{fig:merge}, \ref{fig:graph-rules} and 
through \ref{fig:card-mem-interact} 
in \define{guarded assignment form}.
In such form, the premises of a rule refer to the current state and
the conclusion describes how each state component is changed, if at all, 
by the rule's application.
A derivation rule \define{applies} to a state $\sigma$ 
if all the conditions in the rule's premises hold for $\sigma$ \emph{and} 
the resulting state is different from $\sigma$.
In the rules, we write $S,t$ as an abbreviation for $S \cup \{t\}$.
Rules with two or more conclusions separated by the symbol $\parallel$ 
are non-deterministic branching rules.

%\paragraph{High-level overview.}
The rules are such that it is possible to generate a closed tableau
(or \define{derivation tree}) from an initial state
$\langle \setconst_0, \memconst_0, \arithconst_0, \graphconst_0 \rangle$,
where $\setconst_0$, $\memconst_0$, and $\arithconst_0$ satisfy 
Restriction~\ref{restr} and $\graphconst_0$ is an empty graph,
if and only if %the conjunction of all the constraints in 
$\setconst_0 \cup \memconst_0 \cup \arithconst_0$ is unsatisfiable in $\settheory$.
Broadly speaking, the derivation rules can be divided into three categories.  
First are those that reason about membership constraints 
(of form $x \opin S$). 
These rules only update the components $\setconst$ and $\memconst$ of the current state,
although their premises may depend on other parts of the state,
in particular, the nodes of the graph $\graphconst$. 
Second are rules that handle constraints of the form $c_S \opequal \opcard{S}$. 
The graph incrementally built by the calculus is central to satisfying these constraints. 
Third are rules for propagating element and cardinality constraints,
respectively to $\memconst$ and $\arithconst$.

\begin{figure}[t]
\begin{tabular}{c}
\infrulenew[\ruleUnionDownI]
          {x \opnotin s \opunion t \in \setconstclosed}
          {\setconst := \closedadd{\closedadd{\setconst}{x \opnotin s}}{x \opnotin t}}
\quad
\infrulenew[\ruleUnionDownII]
          {x \opin s \opunion t \in \setconstclosed \\ 
           \{u,v\} = \{s,t\} \\
           x \opnotin u \in \setconstclosed}
          {\setconst := \closedadd{\setconst}{x \opin v}}
%\quad
%\infrulenew[Union Down 3]
%          {x \opin s \opunion t \in \setconstclosed \\ x \opnotin t \in \setconstclosed}
%          {\setconst := \closedadd{\setconst}{x \opin s}}
%          {rule:uniondown3}
\\[4ex]
\infrulenew[\ruleUnionUpI]
          {x \opnotin s \in \setconstclosed \\ x \opnotin t \in \setconstclosed \\ s \opunion t \in \termsof{S} }
          {\setconst := \closedadd{\setconst}{x \opnotin s \opunion t}}
\quad
\infrulenew[\ruleUnionUpII]
          {x \opin u \in \setconstclosed \\
           u \in \{s,t\} \\
           s \opunion t  \in \termsof{S}}
          {\setconst := \closedadd{\setconst}{x \opin s \opunion t}}
%\quad
%\infrulenew[Union Up 3]
%          {x \opin t \in \setconstclosed\\
%           s \opunion t  \in \termsof{S}}
%          {\setconst := \closedadd{\setconst}{x \opin s \opunion t}}
\\[4ex]
\infrulenew[\ruleInterDownI]
          {x \opin s \opinter t \in \setconstclosed}
          {\setconst := \closedadd{\closedadd{\setconst}{x \opin s}}{x \opin t}}{}
\quad
\infrulenew[\ruleInterDownII]
          {x \opnotin s \opinter t \in \setconstclosed \\
           \{u,v\} = \{s,t\} \\
           x \opin u \in \setconstclosed}
          {\setconst := \closedadd{\setconst}{x \opnotin v}}
%\myinfrulenew[Intersection Down 3]
%          {x \opnotin s \opinter t \in \setconstclosed \\
%           x \opin t \in \setconstclosed}
%          {\setconst := \closedadd{\setconst}{x \opnotin s}}
%          {rule:interdown3}
\\[4ex]
\infrulenew[\ruleInterUpI]
          {x \opin s \in \setconstclosed \\
            x \opin t \in \setconstclosed \\
            s \opinter t \in \termsof{S} }
          {\setconst := \closedadd{\setconst}{x \opin s \opinter t}}
\quad
\infrulenew[\ruleInterUpII]
          {x \opnotin u \in \setconstclosed\\
           u \in \{s,t\} \\
           s \opinter t  \in \termsof{S}}
          {\setconst := \closedadd{\setconst}{x \opnotin s \opinter t}}
\\[4ex]
%\infrulenew[Intersection Up 3]
%          {x \opnotin t \in \setconstclosed\\
%           s \opinter t  \in \termsof{S}}
%          {\setconst := \closedadd{\setconst}{x \opnotin s \opinter t}}
%\quad
\infrulenew[\ruleUnionUpSplit]
          {x \opin s \opunion t \in \setconstclosed \\
            x \opin s, % \not\in \setconstclosed \\
            x \opin t \not\in \setconstclosed}
          {\setconst := \closedadd{\setconst}{x \opin s} 
            \ \parallel\ 
            \setconst := \closedadd{\setconst}{x \opin t}}
\\[4ex]
\infrulenew[\ruleInterUpSplit]
          {s \opinter t \in \termsof{S} \\
           \{u,v\} = \{s,t\} \\
            x \opin u \in \setconstclosed \\
          x \opin v, % \not\in \setconstclosed \quad
          x \not\opin v \not\in \setconstclosed}
          {\setconst := \closedadd{\setconst}{x \opin v}
            \ \parallel\
            \setconst := \closedadd{\setconst}{x \not\opin v}}
\\[4ex]
\infrulenew[\ruleSetDifferenceDownI]
          {x \opin s \opsetminus t \in \setconstclosed}
          {\setconst := \closedadd{\closedadd{\setconst}{x \opin s}}{x \opnotin t}}
\quad
\infrulenew[\ruleSetDifferenceDownII]
          {x \opnotin s \opsetminus t \in \setconstclosed \\ x \opin s \in \setconstclosed}
          {\setconst := \closedadd{\setconst}{x \opin t}}
\\[4ex]
\infrulenew[\ruleSetDifferenceDownIII]
          {x \opnotin s \opsetminus t \in \setconstclosed \\ x \opnotin t \in \setconstclosed}
          {\setconst := \closedadd{\setconst}{x \opnotin s}}
\infrulenew[\ruleSetDifferenceUpI]
          {x \opin s \in \setconstclosed \\ x \opnotin t \in \setconstclosed \\ s \opsetminus t \in \termsof{S} }
          {\setconst := \closedadd{\setconst}{x \opin s \opsetminus t}}
\\[4ex]
\infrulenew[\ruleSetDifferenceUpII]
          {x \opnotin s \in \setconstclosed\\
           s \opsetminus t  \in \termsof{S}}
          {\setconst := \closedadd{\setconst}{x \opnotin s \opsetminus t}}
\quad
\infrulenew[\ruleSetDifferenceUpIII]
          {x \opin t \in \setconstclosed\\
           s \opsetminus t  \in \termsof{S}}
          {\setconst := \closedadd{\setconst}{x \opnotin s \opsetminus t}}
\\[4ex]
\infrulenew[\ruleSetDifferenceSplit]
          {s \opsetminus t \in \termsof{S}\\
            x \opin s \in \setconstclosed\\
          x \opin t \not\in \setconstclosed\\
          x \not\opin t \not\in \setconstclosed}
          {\setconst := \closedadd{\setconst}{x \opin t}
            \ \parallel\
            \setconst := \closedadd{\setconst}{x \not\opin t}}
\end{tabular}
\caption{Union, intersection, and set difference rules.}
\label{fig:set-rules1}
\end{figure}

\subsection{Set reasoning rules}\label{sec:algo:mempred}

%\ctsaysi{Text to be revised from here}

Figures~\ref{fig:set-rules1} and~\ref{fig:set-rules2} focus on sets without cardinality.  
They are based on the MLSS decision procedure by Cantone and Zarba~\cite{cantone1998},
though with some key differences. First, the rules operate over a set
$\termsof{S}$ of terms with sort $\setSort$ which may be larger than just the
terms in $\setconst$.  This generalization is required because of
additional terms that may be introduced when reasoning about
cardinalities. Second, the reasoning is done modulo equality. A final,
technical difference is that we work with sets of ur-elements rather
than untyped sets.

\newcommand{\opequalclosedfn}[1]{\opequal_{#1}^*}
\newcommand{\opequalclosed}[3]{#2 \opequalclosedfn{#1} #3}
These rules rely on the following additional notation.  
For any set $\mathcal{C}$ of constraints,
let $\termsofsort{\sigma}{\mathcal{C}}$ refer to terms of
sort $\sigma$ in $\mathcal{C}$, with $\realtermsof{\mathcal{C}}$ denoting all terms
in $\mathcal{C}$.
We define the binary relation
$\opequalclosedfn{\mathcal{C}}\ \subseteq \realtermsof{\mathcal{C}} \times \realtermsof{\mathcal{C}}$
to be the reflexive, symmetric, and transitive closure of the relation on
terms induced by the equality constraints in $\mathcal{C}$.
Now, we define the following closures on the components $\memconst$ and $\setconst$ 
of a state $\langle \setconst, \memconst, \arithconst, \graphconst \rangle$:
\begin{align*}
  \memconstclosed = &
  \setbuilder{x \opequal y}{\opequalclosed{\memconst}{x}{y}}
    \cup \setbuilder{x \not\opequal y}
             {\exists x', y'.~\opequalclosed{\memconst}{x}{x'},
               ~\opequalclosed{\memconst}{y}{y'},
               ~x' \not\opequal y' \in \memconst}\\
\setconstclosed = &~\setconst
   \cup \setbuilder{x \opin s}
             {\exists x', s'.~\opequalclosed{\memconst}{x}{x'},
               ~\opequalclosed{\setconst}{s}{s'},
               ~x' \opin s' \in \setconst}\\
 & \phantom{~\setconst} \cup \setbuilder{x \opnotin s}
             {\exists x', s'.~\opequalclosed{\memconst}{x}{x'},
               ~\opequalclosed{\setconst}{s}{s'},
               ~x' \opnotin s' \in \setconst}
\end{align*}
where $x$, $y$, $x'$, $y'$ in $\termsofsort{\elemSort}{\memconst \cup \setconst}$,
and $s$, $s'$ in $\termsofsort{\setSort}{\setconst}$.
Next, we define a left-associative binary operator $\closedaddfn$
that takes as input a set $\mathcal{C}$ of constraints and a single constraint $l$.
Intuitively,
%given a set of constraints $\mathcal{C}$ and a literal $l$,
$\closedadd{\mathcal{C}}{l}$ adds $l$ to $\mathcal{C}$
only if $l$ is not in $\mathcal{C}$'s closure. More precisely,
\begin{equation}
  \closedadd{\mathcal{C}}{l} =
  \begin{cases}
    \mathcal{C} & \text{if } l \in \mathcal{C}^*\\
    \mathcal{C} \cup \{l\} & \text{otherwise}
  \end{cases}
\end{equation}

The set of \define{relevant} terms, denoted by $\mathcal{T}$,  
for a state $\langle \setconst, \memconst, \arithconst, \graphconst \rangle$
consists of all terms from $\setconst$ and $\graphconst$, namely:
$\realtermsof{\setconst} \cup \vertices$.

Figure~\ref{fig:set-rules1} shows the rules for reasoning about membership in
unions, intersections,
and differences.
Each rule covers one case in which a new
membership (or non-membership) constraint can be deduced.  The justification for these
rules is straightforward based on the semantics of the set operations.  \
The restriction $\{u,v\} = \{s,t\}$ in the premise of some of the rules
cover all the various cases where $s$, say, is the same as $t$, different from $t$, 
the same as $u$, and the same as $v$.
%Due to space limitations, we do not show the rules that process set difference
%constraints.  However, they are analogous to those given for union and
%intersection constraints.  
Figure~\ref{fig:set-rules2} shows rules for
singletons, disequalities, and contradictions.  Note in particular that the
\ruleSetDiseq\ rule introduces a fresh variable $y$, denoting an element
that is in one set but not in the other.

\begin{figure}[t]
\centering
\begin{tabular}{c}
 \infrulenew[\ruleSingle]
    {\opsingleton{x} \in \termsof{\setconst}}
    {\setconst := \closedadd{\setconst}{x \opin \opsingleton{x}}}
%          {rule:singletontrivial}
\quad
\infrulenew[\ruleSingleMem]
    {x \opin \opsingleton{y} \in \setconstclosed}
    {\memconst := \closedadd{\memconst}{x \opequal y}}
%          {rule:singletonin}
\quad
\infrulenew[\ruleSingleNonMem]
    {x \opnotin \opsingleton{y} \in \setconstclosed}
    {\memconst := \closedadd{\memconst}{x \not\opequal y}}
%          {rule:singletonnotin}
\\[4ex]
\infrulenew[\ruleSetDiseq]
          {s \not\opequal t \in \setconst\\
            \setbuilder{x \in \termsof{\setconst}}{x \opin s, x \opnotin t \in \setconstclosed} = \emptyset \\
           \setbuilder{x \in \termsof{\setconst}}{x \opnotin s, x \opin t \in \setconstclosed} = \emptyset \\
          }
          {\setconst := \closedadd{\closedadd{\setconst}{y \opin s}}{y \opnotin t}
            \quad\parallel\quad
            \setconst := \closedadd{\closedadd{\setconst}{y \opnotin s}}{y \opin t}}
%          {rule:setdiseq}
%% \myinfrulenew[Member disequality]
%%           {x \opin s \in \setconstclosed, y \opnotin s \in \setconstclosed}
%%           {\memconst := x \not\opequal y, \memconst}
%%           {rule:memdiseq}
\\[4ex]
\infrulenew[\ruleEqUnsat]
    {(x \not\opequal x) \in \memconstclosed}
    {\unsat}
%          {rule:elemfalse}    
\quad
\infrulenew[\ruleSetUnsat]
    {(x \opin s) \in \setconstclosed\\
      (x \not\opin s) \in \setconstclosed}
    {\unsat}
%          {rule:setfalse}
\quad
\infrulenew[\ruleEmptyUnsat]
    {(x \opin \emptyset) \in \setconstclosed}
    {\unsat}
%          {rule:emptysetfalse}    
\end{tabular}
\caption{Singleton, disequality and contradiction rules.
Here, $y$ is a fresh variable.
}
\label{fig:set-rules2}
\end{figure}

%\subsubsection{Optional propagation rules}
%\myinfrulenew[Singleton and Union]
%          {x \opin \asetv \in \setconstclosed \\
%            \bsetv \opequal \opsingleton{x} \in \setconstclosed \\
%            \asetv \opunion \bsetv \in \termsof{\setconst}}
%          {\setconst := \closedadd{\setconst}{ \asetv \opequal \asetv \opunion \bsetv}}
%          {rule:singandunion}

\begin{exa} \label{ex:setconstraints}
  Let 
  \[ \setconst = \{ S \opequal A
  \opunion B,\, S \opequal C\opinter{D},\, x \opin C,\, x\not\opin D,\, y \not\opin S,\, y \opin D \} .
  \]
  Using the rules in Figure~2, we can directly deduce the additional
  constraints: $x \not\opin C\opinter{D}$ (by \ruleInterUpII), $x\not\opin A$,
  $x\not\opin B$, $y\not\opin A$, $y\not\opin B$ (by \ruleUnionDownI), and $y
  \not\opin C$ (by \ruleInterDownII).  This gives a complete picture, modulo
  equality, of exactly which sets contain $x$ and $y$.
  \qed
\end{exa}

\subsection{Cardinality of sets}\label{sec:algo:card}

The next set of rules, described in Figure~\ref{fig:introduce} and Figure~\ref{fig:merge}, 
operate on the graph component of the current state.
Their purpose is to modify the graph so as to capture the mutual dependencies 
between set and cardinality constraints.
They are based on the observation that
$(i)$ the cardinality of two sets, and that of their union, intersection and set
difference are interrelated; and
$(ii)$ if two set terms are asserted to be equal, their cardinalities must match.

\begin{figure}[t]
\begin{minipage}{.49\textwidth}
\hspace*{.2\textwidth}
\begin{tikzpicture}
  \draw (0,0) circle (3em);
  \draw (3em,0) circle (3em);
  \node at (0,-4em) {$\bsetv$};
  \node at (3em,-4em) {$\csetv$};
  \node at (-1.5em,0) {$\bsetv \opsetminus \csetv$}; 
  \node at (1.5em,0) {$\bsetv \opinter \csetv$};
  \node at (4.5em,0) {$\csetv \opsetminus \bsetv$};
  \node at (1.5em,4em) {$\bsetv \opunion \csetv$};
\end{tikzpicture}
\caption{Venn regions for $\bsetv$ and $\csetv$.}
\label{fig:venn-regions}
\end{minipage}
\begin{minipage}{.49\textwidth}
\vspace*{20pt}
\hspace*{.1\textwidth}
\begin{tikzpicture}[node distance = 1cm]
\node (A) {$\bsetv$};
\node (AuB) [right=of A] {$\bsetv \opunion \csetv$};
\node (B) [right=of AuB] {$\csetv$};
\node (AmB) [below=of A] {$\bsetv \opsetminus \csetv$};
\node (AiB) [below=of AuB] {$\bsetv \opinter \csetv$};
\node (BmA) [below=of B] {$\csetv \opsetminus \bsetv$};
\draw[->] (A) -- (AmB);
\draw[->] (A) -- (AiB);
\draw[->] (AuB) -- (AmB);
\draw[->] (AuB) -- (AiB);
\draw[->] (AuB) -- (BmA);
\draw[->] (B) -- (AiB);
\draw[->] (B) -- (BmA);
\end{tikzpicture}
\vspace*{11pt}
\caption{The same structure as a graph.}
\label{fig:graph-regions}
\end{minipage}
\end{figure}

Figure~\ref{fig:venn-regions} shows the Venn regions for two sets, $\bsetv$ and $\csetv$.
%CT streamlined
%Notice the following relationships: $\bsetv$ is a disjoint union of $\bsetv
%\opsetminus \csetv$ and $\bsetv \opinter \csetv$; $\bsetv \opunion \csetv$ is a
%disjoint union of $\bsetv \opsetminus \csetv$ and $\bsetv \opinter \csetv$ and $\csetv \opsetminus \bsetv$; and
%$\csetv$ is a disjoint union of $\bsetv \opinter \csetv$ and $\csetv \opsetminus \bsetv$.
%Knowing that the sets are \emph{disjoint} is
%important; it allows us to infer the constraints:
The fact that $\bsetv \opsetminus \csetv$, $\bsetv \opinter \csetv$ and 
$\csetv \opsetminus \bsetv$ are disjoint imposes the following relationships 
between their cardinalities and those of $\bsetv$, $\csetv$ and $\bsetv \opunion \csetv$:
\begin{align*}
\opcard{\bsetv} & \,\opequal\, \opcard{\bsetv \opsetminus \csetv} \:\opplus\: \opcard{\bsetv \opinter \csetv}\\
\opcard{\bsetv \opunion \csetv} & \,\opequal\, \opcard{\bsetv \opsetminus \csetv} \:\opplus\: \opcard{\bsetv \opinter \csetv} \:\opplus\: \opcard{\csetv \opsetminus \bsetv}\\
\opcard{\csetv} & \,\opequal\, \opcard{\csetv \opsetminus \bsetv} \:\opplus\: \opcard{\bsetv \opinter \csetv}\ .
\end{align*}
%and for any values $\bsetv \opsetminus \csetv$, $\bsetv \opinter \csetv$ and $\csetv
%\opsetminus \bsetv$ build a model for $\bsetv$, $\csetv$, and $\bsetv \opinter \csetv$ such that
%everything is consistent.% (as long as they can be non-empty).

We can represent these same relationships using the graph 
in Figure~\ref{fig:graph-regions}.  
The nodes of the graph are set terms, and each node has the property of being the
disjoint union of its children in the graph.
Our calculus incrementally constructs a similar graph containing 
all nodes whose cardinality is implicitly or explicitly constrained by the current state. 
%CT remove this because it is talking about an application order
%We start by adding set terms to the
%graph about which there are explicit constraints. Next, we add
%terms whose cardinality is implicitly constrained.  
%As mentioned above, these include
Set terms with implicit cardinality constraints include
$(i)$ union, intersection, and set difference terms appearing in $\setconst$, for
which one of the operands is already in the graph; and 
$(ii)$ terms occurring in an equality whose other member is already in the graph.
A careful analysis\footnote{See completeness proof in~\cite[Chapter 2]{B16} for further details.} reveals 
that we can actually avoid adding intersection terms $\bset \opinter \cset$ unless both 
$\bset$ and $\cset$ are already in the graph, and set
difference terms $\bset \opsetminus \cset$ unless $\bset$ is already in the graph.

%\ksays{R1 complains about $S$ and $\setconst$ being very similar in Figure. likely will fall into wontfix.}
%
The rules in Figure~\ref{fig:introduce} make use of a function $\add$ which takes 
a graph $\graphconst$ and a term $\aset$ and returns the graph $\graphconst'$ 
defined as follows:
\begin{enumerate}
\item For $\aset = \bsetv$ or $\aset = \opemptyset$ or
  $\aset = \opsingleton{x}$:
\begin{itemize}
\item[] $V(\graphconst') = \vertices \cup \{ \aset \}$
\item[] $E(\graphconst') = \edges$
\end{itemize}
\item For $\aset = \bsetv \opinter \csetv$ or $\aset = \bsetv \opsetminus
  \csetv$:
\begin{itemize}
\item[] $V(\graphconst') = V_2 = \vertices \cup \{ \bsetv, \csetv, \bsetv \opsetminus \csetv,
  \bsetv \opinter \csetv, \csetv \opsetminus \bsetv \}$
\item[] $E(\graphconst') = E_2 = \edges \cup \{ (\bsetv, \bsetv \opsetminus \csetv), (\bsetv, \bsetv \opinter \csetv)$,
$(\csetv, \bsetv \opinter \csetv)$, $(\csetv, \csetv \opsetminus \bsetv) \}$
\end{itemize}
\item For $\aset = \bsetv \opunion \csetv$ and $V_2$ and $E_2$ as above:
\begin{itemize}
\item[] $V(\graphconst') = V_2 \cup \{\bsetv \opunion \csetv\}$
\item[] $E(\graphconst') = E_2 \cup \{(\bsetv \opunion \csetv, \bsetv \opsetminus \csetv),
(\bsetv \opunion \csetv, \bsetv \opinter \csetv), (\bsetv \opunion \csetv,
\csetv \opsetminus \bsetv) \}$
\end{itemize}
\end{enumerate}
Recall that, by assumption, each set variable participates in at
most one union, intersection, or set difference in the input set of constraints.
%CT added
It is not difficult to see that this property is preserved by every rule.
This ensures that edges from a set variable node are added to the graph only once,
maintaining the invariant that its children in the graph are disjoint. The only
other rule which adds edges to the graph is the \ruleMergeEqII rule, but it only
adds nodes from the leaves of the graph, creating a new set of disjoint leaves.

Terms with explicit constraints on their cardinality are added to the graph 
by rule \ruleIntroEqCard.  
Terms that have implicit constraints on their cardinality, specifically, singletons and the empty set,
are added by rules \ruleIntroSingleton\ and \ruleIntroEmtpy.

If two nodes $\aset$ and $\bset$ in the graph are explicitly asserted to be equal
(that is, $\aset \opequal \bset \in \setconst$ or $\bset \opequal \aset \in \setconst$), 
we can ensure they have the same cardinality by systematically modifying the graph
as follows.
Let $\nonemptyleaves{n}$ denote the set of leaf nodes for the subtree
rooted at node $n$ which are not known to be empty. Formally,
\begin{equation}
  \nonemptyleaves{n} = \setbuilder{n' \in \leaves{n}}{ n' \opequal \opemptyset \not\in \setconstclosed},
  \label{eq:nonemptyleaves}
\end{equation}
where $\leaves{v} = \setbuilder{w \in \vertices}{\neigh{w} = \emptyset, w \text{ is reachable from } v}$
and $\neigh{w}$ denotes the children of $w$.
We call two nodes $n$ and $n'$ \define{merged} if they have the same set of
nonempty leaves, that is if $\nonemptyleaves{n} = \nonemptyleaves{n'}$.

\begin{figure}[t]
\centering
\begin{tabular}{c}
\infrulenew[\ruleIntroEqRight]
          {\asetv \opequal \bset \in \setconst \\
           \asetv \in \vertices \\
           \bset \not\in \vertices }
          {\graphconst := \add(\graphconst, \bset)}
%          {rule:introduce2}
\quad
\infrulenew[\ruleIntroEqLeft]
          {\asetv \opequal \bset \in \setconst \\
            \asetv \not\in \vertices \\ 
            \bset \in \vertices }
          {\graphconst := \add(\graphconst, \asetv)}
%          {rule:introduce3}
\\[4ex]
\infrulenew[\ruleIntroEqUnion]
          {\asetv \opequal \bsetv \opunion \csetv \in \setconst \\
           \bsetv \opunion \csetv \not\in \vertices \\
           \bsetv\in\vertices \text{ or } \csetv\in\vertices}
          {\graphconst := \add(\graphconst, \bsetv \opunion \csetv) }
%          {rule:introduceunion}
\\[4ex]
\infrulenew[\ruleIntroEqInter]
          {\asetv \opequal \bsetv \opinter \csetv \in \setconst\\
           \bsetv \opinter \csetv \not\in \vertices \\
           \bsetv\in\vertices\\ \csetv\in\vertices}
          {\graphconst := \add(\graphconst, \bsetv \opinter \csetv) }
%          {rule:introduceinter}
\\[4ex]
\infrulenew[\ruleIntroEqSetDiff]
          {\asetv \opequal \bsetv \opsetminus \csetv \in \setconst\\
           \bsetv\in\vertices\\
           \bsetv \opsetminus \csetv \not\in \vertices}
          {\graphconst := \add(\graphconst, \bsetv \opsetminus \csetv) }
%          {rule:introducesetminus}
\\[4ex]
\infrulenew[\ruleIntroEqCard]
          {\cardaset \opequal \opcard{\asetv} \in \setconst}
          {\graphconst := \add(\graphconst, \asetv)}
%          {rule:introduce}
\quad
\infrulenew[\ruleIntroEqSingle]
          {\opsingleton{x} \in \realtermsof{\setconst}}
          {\graphconst := \add(\graphconst, \opsingleton{x})}
%          {rule:singletonintroduce}
\quad
\infrulenew[\ruleIntroEqEmpty]
          {\quad}
          {\graphconst := \add(\graphconst, \opemptyset)}
%          {rule:emptysetintroduce}
\end{tabular}
\caption{Graph extension rules.}
\label{fig:introduce}
\end{figure}

\begin{figure}[t]
\centering
\begin{tabular}{c}
\infrulenew[\ruleMergeEqI]
          {\aset \opequal \bset \in \setconst \\
          \aset, \bset, \opemptyset \in \vertices \\
          \{\cset, \dset\} = \{\aset, \bset\} \\
          \nonemptyleaves{\cset} \subsetneq \nonemptyleaves{\dset}}
          {\setconst := \setbuilder{\aset' \opequal \opemptyset}{\aset' \in \nonemptyleaves{\dset} \setminus \nonemptyleaves{\cset}} \cup \setconst}
%          {rule:mergeeq1}
%\quad
%\infrulenew[\ruleMergeEqII]
%          {\aset \opequal \bset \in \setconst \quad
%          \aset, \bset, \opemptyset \in \vertices \quad
%          \nonemptyleaves{\bset} \subsetneq \nonemptyleaves{\aset}}
%          {\setconst := \setbuilder{\cset \opequal \opemptyset}{\cset \in \nonemptyleaves{\aset} \setminus \nonemptyleaves{\bset}} \cup \setconst}
%%          {rule:mergeeq2}
\\[4ex]
%\infrulenew[\ruleMergeEqIII]
\infrulenew[\ruleMergeEqII]
          {\aset \opequal \bset \in \setconst \\
          \aset, \bset \in \vertices \\
          \nonemptyleaves{\aset} \nsubseteq \nonemptyleaves{\bset} \\
          \nonemptyleaves{\bset} \nsubseteq \nonemptyleaves{\aset}}
          {\graphconst := \merge(\graphconst, \aset, \bset)}
%          {rule:mergeeq3}
\end{tabular}
\caption{Merge rules.}
\label{fig:merge}
\end{figure}

The rules in Figure~\ref{fig:merge} ensure 
that for all equalities over set terms,
the corresponding nodes in the graph are merged.  Consider an equality $s
\opequal t$.  Rule \ruleMergeEqI\ handles the case when
either $\nonemptyleaves{\aset}$ or $\nonemptyleaves{\bset}$ is a proper subset
of the other by constraining the extra leaves in the superset to
be empty.  Rule \ruleMergeEqII\ handles the remaining case where neither is a
subset of the other.
The graph $\graphconst' = \merge(\graphconst,s,t)$ is defined as follows,
where $L_1 = \nonemptyleaves{\aset} \setminus
\nonemptyleaves{\bset}$ and $L_2 = \nonemptyleaves{\bset} \setminus
\nonemptyleaves{\aset}$:
\begin{eqnarray*}
V(\graphconst') & = & \vertices \cup \setbuilder{ l_1 \opinter l_2}{l_1 \in L_1, l_2 \in L_2} \\
E(\graphconst') & = & \edges \cup \setbuilder{(l_1, l_1 \opinter l_2),(l_2, l_1 \opinter l_2)}{l_1 \in L_1, l_2 \in L_2}
\end{eqnarray*}
\ruleMergeEqII\ introduces a quadratic number of leaves
($\card{L_1} \cdot \card{L_2}$). To reduce the impact of merge operations, a
useful rule to apply early on is the (optional) \ruleGuessEmpty\ rule
in Figure~\ref{fig:graph-rules}. It guesses if a leaf node is equal to
the empty set or not. The use of this rule is illustrated in
Example~\ref{ex:setcardinality}.
Here and in Figure~\ref{fig:card-mem-interact},
$\leaves{\graphconst} = \{v\in\vertices\ |\ \neigh{v}=\emptyset\}$.

%%%% no longer needed with new normal form --Kshitij [Nov 24, 2015]
%% The next rule is to merge different occurences of the same set
%% variable in the graph:
%% \myinfrulenew[Merge Same Set]
%%           {T^U \in \vertices \\ T^{U'} \in \vertices}
%%           {\graphconst := \merge(\graphconst, T^{U}, T^{U'})}
%%           {rule:mergesameset}

%\paragraph{Induced arithmetic constraints.}

%\pagebreak
In Figure~\ref{fig:graph-rules}
we denote by $\grapharithconst$ the collection of all of the following cardinality constraints
imposed by graph $\graphconst$:
\begin{enumerate}
\item For each set term $\aset \in V(\graphconst)$,
  its cardinality (denoted by its corresponding cardinality variable $\cardaset$)
  is the sum of the cardinalities of its non-empty leaf nodes:
  \[\setbuilder{ \cardaset \opequal \sum_{\bset \in \nonemptyleaves{s}} \cardbset }{ \aset \in \vertices }\]
  %% (Optimization: this is sufficient to do only for nodes which
  %% appear in input constraints, and the nodes in their subtree.)
\item Each cardinality is non-negative:
  \[ \setbuilder{ \cardaset \opgeq 0 }{ \aset \in \vertices } \]
  %% (Again, same optimization)
\item Every singleton set has cardinality $1$: 
   \[ \setbuilder{ \cardaset \opequal 1}{ \aset \in \vertices,~\aset = \opsingleton{x} } \]
\item The empty set has cardinality $0$:
  \[ \setbuilder{ \cardaset \opequal 0}{ \aset \in \vertices,~\aset = \opemptyset } \]
\end{enumerate}

\reviewerbis{
You may insist a bit more on the fact that MERGE EQUALITY II is a rule that is
important to pay attention too, if only because its quadratic nature is critical
for performance, it seems. Also indicating that the rule is illustrated in the
next example, Example 3.3.

For instance, it may appear counter-intuitive to prefer applying branching
rules, namely those of $R_1$ or GUESS EMPTY SET before applying the non-branching
rule of merging (as we see for instance in section 5.1).

You may also emphasize a bit more, when presenting GUESS LOWER BOUND, that the
rule is optional, i.e. not needed for completeness. It is said later but the
reader might enjoy the remark when it is presented.
}
Rule \ruleAriContra\  %, shown in Figure~\ref{fig:graph-rules} 
relies on the arithmetic solver to check 
whether the constraints in $\grapharithconst$ are inconsistent with the input cardinality constraints.  

\begin{figure}[t]
\centering
\begin{tabular}{c}
\infrulenew[\ruleAriContra]
    {\arithconst \cup \grapharithconst \models_{\ariththeory} \bot}
    {\unsat}
%          {rule:arithfalse}    
\quad
\infrulenew[\ruleGuessEmpty]
          {t \in \leaves{\graphconst}}
          {\setconst := \closedadd{\setconst}{t \opequal \emptyset}
           \quad\parallel\quad
           \setconst := \closedadd{\setconst}{t \not\opequal \emptyset}}
%          {rule:guessemptyset}
\end{tabular}
\caption{Additional graph rules.}
\label{fig:graph-rules}
\end{figure}

\subsection{Cardinality and membership interaction}\label{sec:algo:interaction}
The rules in Figure~\ref{fig:card-mem-interact}
propagate consequences of set membership constraints to the state components
$\memconst$ and $\arithconst$.
Let $\eqls$ denote the set of equalities in $\memconst$, and let
$\equivcl{\eqls}{x}$ denote the equivalence class of $x$ with respect to $\eqls$. 
In the rules, for term $t$ of sort $\setSort$,
$t_\setconst$ denotes the set
$\setbuilder{\equivcl{\eqls}{x}}{x \opin t \in \setconstclosed}$
of equivalence classes of elements known to be in $t$.
The notation $\arithconst \triviallyimply c_t \geq n$ means that
$c_t \opgeq k \in \arithconst$ for some concrete constant $k \geq n$.

\begin{figure}[t]
%\centering
\begin{tabular}{c}
\infrulenew[\ruleMemArrange]
    {t \in \allleaves \quad
      \arithconst \nottriviallyimply c_t \geq \card{t_\setconst} \quad
      \equivcl{\eqls}{x}, % \in t_\setconst \\
      \equivcl{\eqls}{y} \in t_\setconst \quad
      \equivcl{\eqls}{x} \neq \equivcl{\eqls}{y} \quad
      x \not\opequal y \not\in \memconstclosed
    }
    {\memconst := \closedadd{\memconst}{x \opequal y}
     \quad\parallel\quad
     \memconst := \closedadd{\memconst}{x \not\opequal y}}
%          {rule:memarrange}
      %% x, y \in \termsof{\memconst}\\
      %% x \opequal y \not\in \memconstclosed\\
\\[4ex]
\infrulenew[\ruleGuessLB]
    %% {x_1 \opin \aset, x_2 \opin \aset, \ldots, x_n \opin \aset \in \setconstclosed\\
    %%   x_i \opequal x_j \not\in \memconst^* \text{ for all } 1 \leq i \neq j \leq n}
    {\bset \in \allleaves \\
      \arithconst \nottriviallyimply \cardbset \geq \card{\bset_\setconst}
      \\
      \cardbset \opless \card{\bset_\setconst} \not\in \arithconst
    }
    {\arithconst := \cardbset \opgeq \card{\bset_\setconst}, \arithconst
     \quad\parallel\quad 
     \arithconst := \cardbset \opless \card{\bset_\setconst}, \arithconst}
%          {rule:guesslowerbound}
\\[4ex]
%\quad
\infrulenew[\rulePropMinsize]
    {x_1 \opin \aset, \ldots, x_n \opin \aset \in \setconstclosed \quad
     \arithconst \nottriviallyimply \cardaset \geq n
     \\
     x_i \not\opequal x_j \in \memconst^* \text{ for all } 1 \leq i < j \leq n
    }
    {\arithconst := \cardaset \opgeq n, \arithconst}
%          {rule:propminsize}
\end{tabular}
\caption{
Cardinality and membership interaction rules.
}
\label{fig:card-mem-interact}
\end{figure}

Rule \ruleMemArrange\ is used to decide which element variables constrained
to be in the same set $t$ should be identified and which should not.
Once applied to completion, Rule \rulePropMinsize\ can then be used
to determine a lower bound for the cardinality of that set.
The (optional) rule \ruleGuessLB\ can be used to short-circuit this process by guessing a
conservative lower bound based on the number of distinct equivalence classes of
elements known to be members of a set.  If this does not lead to a
contradiction, a model can be found without resorting to an extensive use of
\ruleMemArrange.

%% \begin{equation}
%% \myinfrule[Two Must be Equal]
%%     {x_1 \opin \aset, \ldots, x_n \opin \aset \in \setconstclosed\\
%%       x_i \opequal x_j \not\in \memconst^* \text{ for } 1 \leq i \neq j \leq n \\
%%       \cardaset < n \in \arithconst
%%     }
%%     {\memconst := x_1 \opequal x_2, \memconst\\
%%       \memconst := x_1 \opequal x_3, \memconst\\ 
%%       \ldots \\
%%       \memconst := x_{n-1} \opequal x_n, \memconst}
%% \end{equation}

\reviewer{ It would have been nice to see actually built tableaux in
  the examples. (Optional: reviewer says can be skipped if size is an
  issue.)  }
\begin{exa}
  \label{ex:setcardinality}
  Consider again the constraints from Example~\ref{ex:setconstraints}, but now
  augmented with cardinality constraints:
  \begin{eqnarray*}
   \setconst & = & \{ S \opequal A \opunion B,\, S \opequal C\opinter{D},\, x \opin C,\, x\not\opin D,\, y \not\opin S,\, y \opin D \} \\
   \arithconst & = & \{c_S \opequal \opcard{S},\, c_C \opequal \opcard{C},\, c_D \opequal \opcard{D},\, c_S \opgeq 4,\, c_C + c_D \opless 10\}
  \end{eqnarray*}
  Using the rules in Figure~6, the
  following nodes get added to the graph: $S$, $C$, $D$ (by \ruleIntroEqCard), $A
  \opunion B$, $C \opinter D$ (by \ruleIntroEqRight). Node $A \opunion B$ is added
  with children $A \opsetminus B$, $A \opinter B$, and $B \opsetminus A$; and
  by adding $C \opinter D$, we also get $C \opsetminus D$ and $D \opsetminus
  C$, with the corresponding edges from $C$ and $D$.  Now, using two
  applications of \ruleMergeEqII, we force the sets $S$, $A \opunion B$ and $C
  \opinter D$ to have the same set of 3 leaves, labeled $S \opinter (A
  \opsetminus B) \opinter (C \opinter D)$, $S \opinter (A \opinter B) \opinter
  (C \opinter D)$, and $S \opinter (B \opsetminus A) \opinter (C \opinter D)$.
  Let us call the latter nodes respectively $l_1$, $l_2$, and $l_3$, for convenience.  
  Let us also
  designate $l_4 = C \opsetminus D$ and $l_5 = D \opsetminus C$.  Notice that
  the induced cardinality constraints now include $c_S \opequal c_{l_1} \opplus c_{l_2} \opplus
  c_{l_3}$, $c_C \opequal c_{l_1} \opplus c_{l_2} \opplus c_{l_3} \opplus c_{l_4}$, and $c_D \opequal c_{l_1} \opplus
  c_{l_2} \opplus c_{l_3} \opplus c_{l_5}$.  With the addition of $C \opsetminus
  D$ and $D \opsetminus C$ to the graph, these are also added to
  $\termsof{\setconst}$.  We can then deduce $x \opin C \opsetminus D$ and $y
  \opin D \opsetminus C$ using the rules for set difference.  
  Finally, we can use \rulePropMinsize\ to deduce $c_{l_4} \opgeq 1$
  and $c_{l_5} \opgeq 1$.  It is now not hard to see that using pure arithmetic
  reasoning, we can deduce that $c_C \opplus c_D \opgeq 10$ which leads to \unsat
  using \ruleAriContra.
  \qed
\end{exa}

%%% Local Variables:
%%% mode: latex
%%% TeX-master: "pap.tex"
%%% End:

%!TEX root =  pap.tex

\section{Calculus Correctness}
\label{sec:sets:correct}

Our calculus is terminating and sound for any derivation strategy, 
that is, regardless of how the rules are applied.
It is also refutation complete for any \define{fair} strategy, defined as a strategy
that does not delay indefinitely the application of an applicable derivation rule.
%For space reasons, we only outline the proof arguments here.  Complete proofs are
%given in~\cite{B16}.

To prove these properties it is convenient
to partition the derivation rules of the calculus in the following subsets.
\begin{itemize}
\item[] $\mathcal{R}_1$, membership predicate reasoning rules, from Figures 2 and 3.
\item[] $\mathcal{R}_2$, graph rules to reason about cardinality, from Figures 6, 7 and 8.
\item[] $\mathcal{R}_3$, rules from Figure 9 other than Rule \ruleGuessLB.
%\item $\mathcal{R}_4$: optional propagation and split rules other than Rule \ruleGuessLB.
\item[] $\mathcal{R}_4$, rule \ruleGuessLB.
\end{itemize}

The rules are used to construct derivation trees.
A \define{derivation tree} is a tree over states,
with a root of the form
$\langle \setconst_0, \memconst_0, \arithconst_0, (\emptyset,\emptyset) \rangle$
where $\setconst_0 \cup \memconst_0 \cup \arithconst_0$ satisfies Restriction~\ref{restr}
and the children of each non-root node are obtained by applying one 
of the derivation rules of the calculus to that node.
%We call the root of a derivation tree an \define{initial} state.
Let $\mathcal{R}$ be a subset of the derivation rules of the calculus. 
A state is \define{saturated} with respect to $\mathcal{R}$ 
if no rules in $\mathcal{R}$ apply to it.
A branch of a derivation tree is \define{closed} if it ends with \unsat;
it is \define{saturated} with respect to $\mathcal{R}$ if so is its leaf.
A derivation tree is \define{closed} if all of its branches are closed.
A derivation tree \define{derives} from a derivation tree $T$ 
if it is obtained from $T$ by the application of exactly one of the derivation rules 
to one of $T$'s leaves.

\begin{defi}[Derivations]\em
Let $\mathcal{C}$ be a set of $\settheory$-constraints.
A \define{derivation (of $\mathcal{C}$)} is a sequence $(T_i)_{0 \leq i \leq \kappa}$ of derivation trees,
with $\kappa$ finite or countably infinite,
such that $T_{i+1}$ derives from $T_i$ for all $i$,
and
$T_0$ is a one-node tree whose root is a state
$\langle \setconst_0, \memconst_0, \arithconst_0, (\emptyset,\emptyset)\rangle$
where $\setconst_0 \cup \memconst_0 \cup \arithconst_0 = \mathcal{C}$.
A \define{refutation (of $\mathcal{C}$)} is a (finite) derivation of $\mathcal{C}$ that ends with a closed tree.
\end{defi}

\begin{rem}
In the proofs below we implicitly rely on the fact that,
for every state $\langle \setconst, \memconst, \arithconst, \graphconst \rangle$
in a derivation tree, the constraints in $\setconst \cup \memconst \cup \arithconst$
satisfy Restriction~\ref{restr}.
This is the case because the restriction is imposed on root states
and is preserved by all of its rules, as one can easily verify.
\end{rem}

%------------------------------------------------------------------------------
\subsection{Termination}
%------------------------------------------------------------------------------
%% Note that rules in Section \ref{sec:algo:mempred} do not introduce new
%% terms (except the set disequality rule, but that can be applied only
%% finite number of times). As for rules is Section \ref{sec:algo:card},
%% note that the each of the rule once applied is no longer applicable.
%% New terms are only added by $\merge$, but once
%% a merge operation has happened, the nodes stay merged. If there are
%% only finite number of equalities between set terms, there can be only
%% finite number of $\merge$ applications.
%% Note that only rule that adds equality
%% between set terms is the optional propagation rule (which can also be
%% applied only a finite number of times). Rules in Section
%% \ref{sec:algo:interaction} either guess equality between element terms
%% (finitely many times for each set term),
%% or add an arithmetic constraint (and the premise of the rule no longer holds).
%% These rule also do not introduce any
%% $\setSort$ or $\elemSort$ terms.
%% It follows that each rule can be applied
%% only a finite number of times,
%% and hence the procedure eventually terminates.
%% We formalize the reasoning below.
%
\begin{prop}[Termination] \label{prop:finite-derivations}
%Let $\setconst_0$, $\memconst_0$, and $\arithconst_0$ be normalized
%set, element, and arithmetic respectively.
%
Let $\mathcal{R}$ collect all rules in our calculus except for (the optional) rule
%% \ref{rule:singandunion} and
%\ref{rule:guesslowerbound}.
\ruleGuessLB.
Every derivation %$\derivation$
using only rules from $\mathcal{R}$ is finite.
\end{prop}
\begin{proof}
%% Commands
\newcommand{\wfra}[1]{f_1(#1)} 
\newcommand{\wfre}[1]{f_2(#1)}
\newcommand{\wfri}[1]{f_3(#1)} 
\newcommand{\wfrb}[1]{f_4(#1)} 
\newcommand{\wfrc}[1]{f_5(#1)}
\newcommand{\wfrd}[1]{f_6(#1)}
\newcommand{\wfrf}[1]{f_7(#1)} 
\newcommand{\wfrg}[1]{f_8(#1)} 
\newcommand{\wfrh}[1]{f_9(#1)}
\newcommand{\wfrmax}{9}
Let 
$\langle \setconst_0, \memconst_0, \arithconst_0, (\emptyset, \emptyset) \rangle$ 
be the initial state of the derivation.
We first define a well-founded relation $\termreln$ over states.
Next, we show that application of any rule in $\mathcal{R}$ 
to a leaf of a derivation tree gives smaller states with respect to this relation. 
As the relation is well-founded, it will follow
that the derivation cannot be infinite.

In order to define $\termreln$, we define $f_i$ for $i \in \{1, 2, \ldots, \wfrmax\}$,
each of which maps a state
$\sigma = \langle \setconst, \memconst, \arithconst, \graphconst \rangle$
to a natural number (non-negative integer). We denote the set of natural numbers by $\mathbb{N}$.
\begin{itemize}
\item $\wfra{\sigma}$: number of equalities $t_1 \opequal t_2$ in $\setconst$ such
  that either $t_1 \not \in \vertices$, $t_2 \not\in \vertices$, or
  $\nonemptyleaves{t_1} \neq \nonemptyleaves{t_2}$.
\item $\wfre{\sigma}$: size of $(\termsofsort{\setSort}{\setconst} \cup \{\opemptyset\}) \setminus \vertices$.
\item $\wfri{\sigma}$: size of $\setbuilder{t \in \leaves{\graphconst}}{t \opequal \opemptyset \not\in \setconstclosed, t \not\opequal \opemptyset \not\in \setconstclosed}$.
\item $\wfrb{\sigma}$: number of disequalities $t_1 \not\opequal t_2$ in $\setconst$ such
  that the premise of \ruleSetDiseq\ holds.
\item $\wfrc{\sigma}$: size of $\termsofsort{\setSort}{\setconst} \cup \{\opemptyset\}\cup \vertices$.
\item $\wfrd{\sigma}$: size of $\termsofsort{\elemSort}{\setconst \cup \memconst}$.
\item $\wfrf{\sigma}$: size of $\memconstclosed$ subtracted from $2 \cdot \inparen{\wfrd{\sigma}}^2$.
  As all constraints in $\memconstclosed$ are either $x \opequal y$ or $x \not\opequal y$ with $x$ and $y$ in
  $\termsofsort{\elemSort}{\setconst \cup \memconst}$, the size of $\memconstclosed$ can be at most $2 \cdot \inparen{\wfrd{\sigma}}^2$.
  Thus, $\wfrf{\cdot}$ is well-defined as a map into $\mathbb{N}$.
\item $\wfrg{\sigma}$: size of $\setconstclosed$ subtracted from $2 \cdot \inparen{\wfrc{\sigma}}^2 + 2 \cdot \wfrc{\sigma} \cdot \wfrd{\sigma}$.
  There are at most $2 \cdot \inparen{\wfrc{\sigma}}^2$ constraints of the form $s \opequal t$ or $s \not\opequal t$ in $\setconstclosed$ as
  $s$ and $t$ are in $\termsofsort{\setSort}{\setconst} \cup \{\opemptyset\}\cup \vertices$.
  There are at most $2 \cdot \wfrc{\sigma} \cdot \wfrd{\sigma}$ constraints of the form $x \opin s$ or $x \opnotin s$ in $\setconstclosed$ as
  $x$ and $s$ are in $\termsofsort{\elemSort}{\setconst \cup \memconst}$ and $\termsofsort{\setSort}{\setconst} \cup \{\opemptyset\}\cup \vertices$
  respectively. Thus, $\wfrg{\cdot}$ is well-defined as a map into $\mathbb{N}$.
\item $\wfrh{\sigma}$: size of $\inparen{\termsofsort{\setSort}{\setconst} \cup \{\opemptyset\}\cup \vertices} \setminus \setbuilder{t \in \leaves{\graphconst}}{\arithconst \nottriviallyimply \cardbset \geq \opcard{\bset_\setconst}}$.
\end{itemize}
Then, we define the order $\termreln$ over states as follows:
\begin{itemize}
\item $\sigma \termreln \sigma'$ if $\sigma \neq \unsat$ and $\sigma' = \unsat$.
\item $\sigma \termreln \sigma'$ if $\sigma \neq \unsat$, $\sigma' \neq \unsat$, and
  \[
  \inparen{\wfra{\sigma}, \ldots, \wfrh{\sigma}}
  >^\wfrmax_{\textsf{lex}}
  \inparen{\wfra{\sigma'}, \ldots, \wfrh{\sigma'}}
  \]
  where $\inparen{\mathbb{N}^\wfrmax, >^\wfrmax_{\textsf{lex}}}$ is the $9$-fold lexicographic product of
  ordering over natural numbers $\inparen{\mathbb{N}, >}$.
\item $\sigma \not\termreln \sigma'$ otherwise.
\end{itemize}
The well-foundedness of $\termreln$ over states follows from the well-foundedness of
$\inparen{\mathbb{N}^9, >^\wfrmax_{\textsf{lex}}}$ \cite[Section 2.4]{termrewriting98}. 
%% Let $\sigma \termreln \sigma'$ if and only if
%% $(\wfra{\sigma}, \ldots, \wfrh{\sigma}) >_{\textsf{lex}} (\wfra{\sigma'}, \ldots, \wfrh{\sigma'})$

Let $r \in \mathcal{R}$ be a rule applicable at state $\sigma$, and let $\sigma'$ be
the state after the application of the rule (if there are multiple
conclusions, denote the state on first branch as $\sigma_1'$, second branch as $\sigma_2'$ and so on).
We note below for each rule $r \in \mathcal{R}$ the relation between
$\wfra{\sigma}$, $\ldots$, $\wfrh{\sigma}$ and
$\wfra{\sigma'}$, $\ldots$, $\wfrh{\sigma'}$
which establishes that $\sigma \termreln \sigma'$.
\begin{itemize}
\item First, we consider
  Rules for intersection (Figure \ref{fig:set-rules1}),
  union (Figure \ref{fig:set-rules1}),
  set difference (Figure \ref{fig:set-rules1})
  and Rule \ruleSingle~for singleton.
  None of these rules introduce equalities of set terms,
  nor do they affect the graph $\graphconst$; thus
  $\wfra{\sigma} \geq \wfra{\sigma'}$.
  The only terms introduced to $\setconst$ are from $\vertices$, thus
  $\wfre{\sigma} = \wfre{\sigma'}$.
  None of these rules update $\graphconst$ or introduce equalities or disequalities of set terms, thus
  $\wfri{\sigma} = \wfri{\sigma'}$.
  None of these rules introduce disequalities between set terms, thus
  $\wfrb{\sigma} \geq \wfrb{\sigma'}$.
  None of these rules introduce set terms not already in $\setconst$ or $\vertices$, thus
  $\wfrc{\sigma} = \wfrc{\sigma'}$.
  None of the rules introduce $\elemSort$ variables not already in $\setconst$ or $\memconst$, thus
  $\wfrd{\sigma} = \wfrd{\sigma'}$.
  None of these rules update $\memconst$, thus
  $\wfrf{\sigma} = \wfrf{\sigma'}$.

  Each of these rules updates $\setconst$.
  Recall that for a rule to be applicable at $\sigma$,
  the resulting state must be different from $\sigma$.
  From the definition of $\closedaddfn$, we can conclude that the size of $\setconstclosed$ has increased.
  As $\wfrc{\sigma} = \wfrc{\sigma'}$ and $\wfrd{\sigma} = \wfrd{\sigma'}$, it follows that
  $\wfrg{\sigma} > \wfrg{\sigma'}$.
\item Next, we consider Rules \ruleSingleMem,
  \ruleSingleNonMem~and \ruleMemArrange.
  None of these rules introduce equalities of set terms, thus
  $\wfra{\sigma} \geq \wfra{\sigma'}$.
  None of these rules introduce set terms to $\setconst$ or $\vertices$, thus
  $\wfre{\sigma} = \wfre{\sigma'}$.
  None of these rules update $\graphconst$ or introduce equalities or disequality of set terms, thus
  $\wfri{\sigma} = \wfri{\sigma'}$.
  None of these rules introduce disequalities of set terms, thus
  $\wfrb{\sigma} \geq \wfrb{\sigma'}$.
  None of these rules introduce set terms to $\setconst$ or $\vertices$, thus
  $\wfrc{\sigma} = \wfrc{\sigma'}$.
  None of the rules introduce $\elemSort$ variables not already in $\setconst$ or $\memconst$, thus
  $\wfrd{\sigma} = \wfrd{\sigma'}$.
  
  Each of these rules updates $\memconst$.
  From the definition of $\closedaddfn$, we can conclude that the size of $\memconstclosed$ has increased.
  As $\wfrd{\sigma} = \wfrd{\sigma'}$,
  we can conclude that $\wfrf{\sigma} > \wfrf{\sigma'}$.
\item Next, we consider Rule \ruleSetDiseq.
  The rule does not introduce any equality of set terms, thus
  $\wfra{\sigma} \geq \wfra{\sigma_i'}$ for $i \in \{1, 2\}$.
  The rule does not introduce set terms to $\setconst$ or $\vertices$, thus
  $\wfre{\sigma} = \wfre{\sigma_i'}$ for $i \in \{1, 2\}$.
  The rule does not update $\graphconst$,
  thus $\wfri{\sigma} \geq \wfri{\sigma_i'}$ for $i \in \{1, 2\}$.
  The premise of the rule
  does not hold after application of the rule on either of the
  branches. It follows that $\wfrb{\sigma} > \wfrb{\sigma_i'}$
  for $i \in \{1,2\}$.
\item Next, we consider Introduce Rules (Figure 6).
  Note that none of these rules introduce equalities of set terms.
  Also note that if $t_1 \in \vertices$,
  $t_2 \in \vertices$ and $t_1 \opequal t_2$ in $\setconst$ then
  $\nonemptyleaves{t_1} = \nonemptyleaves{t_2}$
  (see Proposition \ref{prop:graph}, property \ref{prop:graph:propequal}).
  Thus, $\wfra{\sigma} \geq \wfra{\sigma'}$.

  Each of the rules adds at least one new node to $\graphconst$ which is in
  $\termsofsort{\setSort}{\setconst} \cup \{\opemptyset\}$. At the same time,
  $\setconst$ is unchanged.
  It follows that
  $\wfre{\sigma} > \wfre{\sigma'}$.
  
  %% As they do not introduce do
  %% not affect $\wfra{\cdot}$, $\ldots$, $\wfrd{\cdot}$. In particular,
  %% $\wfrc{\sigma} = \wfrc{\sigma'}$ as nodes introduced to the graph
  %% are always either in $\termsofsort{\setSort}{\setconst}$ or the
  %% $\emptyset$ symbol. As $\add$ operation adds vertices to the graph,
\item Rules \ruleMergeEqI\ and \ruleMergeEqII.
  Though these rules add equalities of the form $u \opequal \opemptyset$
  to $\setconst$, the equalities are such that $u \in \vertices$, $\opemptyset \in \vertices$ and
  $\nonemptyleaves{u} = \emptyset =\nonemptyleaves{\opemptyset}$.
  %% If $\opemptyset \not\in \vertices$, it trivially follows that $\wfra{\sigma} \geq \wfra{\sigma'}$.
  %% If $\opemptyset \in \vertices$, since $\nonemptyleaves{u} = \emptyset = \nonemptyleaves{\opemptyset}$,
  It follows that $\wfra{\sigma} \geq \wfra{\sigma'}$.
  
  Now, observe that for
  Rule \ruleMergeEqI\ or Rule \ruleMergeEqII\ to be applicable, there must exist $s \opequal t \in \setconst$
  such that $\nonemptyleaves{s} \neq \nonemptyleaves{t}$.
  After the application of the rule, $\nonemptyleaves{s} = \nonemptyleaves{t}$.
  This shows that $\wfra{\sigma} > \wfra{\sigma'}$.
\item Rule \ruleMergeEqIII.
  For the rule to be applicable, there must exist $s \opequal t \in \setconst$
  such that $\nonemptyleaves{s} \neq \nonemptyleaves{t}$.
  After the application of the rule, $\nonemptyleaves{s} = \nonemptyleaves{t}$.
  Thus, necessarily $\wfra{\sigma} > \wfra{\sigma'}$.
\item Rule \ruleGuessEmpty.
  Note that though this rule may add an equality of the form $t \opequal \opemptyset$ on the first branch,
  using the same reasoning as for Rules \ruleMergeEqI\ and \ruleMergeEqII\ above, we can conclude that
  $\wfra{\sigma} \geq \wfra{\sigma_1'}$.
  On the second branch, as no disequality is added, we get that
  $\wfra{\sigma} \geq \wfra{\sigma_2'}$.
  Only terms introduced to $\setconst$ are from $\vertices$, thus
  $\wfre{\sigma} = \wfre{\sigma_1'}$ for  $i \in \{1, 2\}$.

  In order to apply the rule, we pick a $t \in \leaves{G}$ such that
  $t \opequal \opemptyset \not\in \setconstclosed$ and
  $t \not\opequal \opemptyset \not\in \setconstclosed$.
  On the first branch, $t \opequal \opemptyset \in \setconstclosed$, thus $\wfri{\sigma} > \wfri{\sigma_1'}$.
  On the second branch, $t \not\opequal \opemptyset \in \setconstclosed$, thus $\wfri{\sigma} > \wfri{\sigma_2'}$.
\item Rule \rulePropMinsize.
  The rule does not update $\setconst$, $\memconst$, or $\graphconst$, thus
  $\wfra{\sigma} = \wfra{\sigma'}$,
  $\wfre{\sigma} = \wfre{\sigma'}$,
  $\wfri{\sigma} = \wfri{\sigma'}$,
  $\wfrb{\sigma} = \wfrb{\sigma'}$,
  $\wfrc{\sigma} = \wfrc{\sigma'}$,
  $\wfrd{\sigma} = \wfrd{\sigma'}$,
  $\wfrf{\sigma} = \wfrf{\sigma'}$, and
  $\wfrg{\sigma} = \wfrg{\sigma'}$.
  But, $\wfrh{\sigma} > \wfrh{\sigma'}$.
\item Rules \ruleEqUnsat, \ruleSetUnsat, \ruleEmptyUnsat,
  and \ruleAriContra.
  For each of these rules to be applicable, $\sigma \neq \unsat$.
  On the other hand, $\sigma' = \unsat$ after the application of the rule.
  By definition, $\sigma \termreln \sigma'$. \qedhere
\end{itemize}

%CT Already said this
%As $\termreln$ order over states is well-founded, and application of
%any rule gives a smaller state with respect to this order, the
%derivation $\derivation$ must be necessarily finite.
\end{proof}

\begin{rem}
  It is easy to extend the termination proof above to include the optional
  rule %s, like% Rules \ref{rule:singandunion} and
  \ruleGuessLB.  It would involve tracking sizes of
  additional objects%
%
% For instance, for Rule \ref{rule:singandunion},
%  tracking the number of singleton set terms for which the rule has
%  not yet been applied. Since no rule introduces new singleton set
%  terms, this can only go down.
 %% For Rule \ruleGuessLB,
---a strategy similar to the one adopted for Rule \ruleGuessEmpty\ 
  in our proof would suffice.
\end{rem}

\subsection{Completeness}%
%------------------------------------------------------------------------------

%
We prove properties about different subsets of rules,
developing the completeness proof in stages.
%% We develop the completeness proof in stages,
%% proving properties about different subsets of rules.
%
We start with a proposition about rule set $\mathcal{R}_1$.

%% \begin{prop} \label{prop:membermodel}
%%   %CT Streamlined the statement of the proposition
%% Let $\langle \setconst, \memconst, \arithconst, \graphconst\rangle$ be a state
%% to which none of rules in $\mathcal{R}_1$ apply.
%% %CT modified the statement of the proposition because it implicitly posited
%% % the existence of the structure $\structureS$ with the given properties
%% % (this should be proven first).
%% There is a model $\structureS$ of $\settheory$ that satisfies the constraints
%% $\setconst$ and $\memconst$ and has the following properties.
%% \begin{enumerate}
%% \item
%% For all $x, y \in \varsof(\memconst) \cup \varsof(\setconst)$ of sort $\elemSort$,
%% $x^\structureS = y^\structureS$ if and only if $x \opequal y \in \memconstclosed$.
%% \item
%% For all $S \in \varsof(\setconst)$ of sort $\setSort$,
%% $S^\structureS = \setbuilder{ x^\structureS }{ x \opin S \in \setconstclosed}$.
%% \item
%% For all $c_\asetv \in \varsof(\setconst)$ of sort $\cardSort$,
%% $c_\asetv^\structureS = \card{ S^\structureS }$.
%% \end{enumerate}
%% \end{prop}

%!TEX root = pap.tex
\begin{prop}\label{prop:membermodel}
Let $\langle \setconst, \memconst, \arithconst, \graphconst \rangle$
be a derivation tree leaf that is saturated with respect to $\mathcal{R}_1$.
There is a model $\structureS$ of $\settheory$ that satisfies the constraints
$\setconst$ and $\memconst$ and has the following properties.
\begin{enumerate}
\item
For all $x, y \in \varsof(\memconst) \cup \varsof(\setconst)$ of sort $\elemSort$,\\
$x^\structureS = y^\structureS$ if and only if $x \opequal y \in \memconstclosed$.
\item
For all $S \in \varsof(\setconst)$ of sort $\setSort$,
$S^\structureS = \setbuilder{ x^\structureS }{ x \opin S \in \setconstclosed}$.
\item
For all $c_\asetv \in \varsof(\setconst)$ of sort $\cardSort$,
$c_\asetv^\structureS = \card{ S^\structureS }$.
\end{enumerate}
  
\end{prop}
\begin{proof}
Since the models of $\settheory$ are closed under variable reassignment,
we pick an arbitrary model $\structureS$ of $\settheory$ and show that we can change
its interpretation of the variables of $\setconst \cup \memconst$
to satisfy the properties above.

We start by interpreting 
all variables of $\elemSort$ sort in $\setconst \cup \memconst$ so that,  
for all $x$ and $y$ in $\varsof(\memconst) \cup \varsof(\setconst)$ of $\elemSort$
  sort,
  \[ x^\structureS = y^\structureS \text{ if and only } x \opequal y \in \memconstclosed \ \text{.}\]
  %
  %% TODO: ^^ say why this is well-defined / exists
  %
  It follows that $\structureS$ satisfies $\memconst$.
  Next, let $\structureS$ interpret each variable $S$ of $\setSort$ sort in $\varsof(\setconst)$ as:
  \[ S^\structureS = \setbuilder{ x^\structureS }{ x \opin S \in \setconstclosed} \]
  and each variable $c_\asetv$ of $\cardSort$ sort in $\varsof(\setconst)$ as:
  \[ c_\asetv^\structureS = \card{ S^\structureS } \ \text{.}\]
  %
  % TODO: ^^ say why this is well-defined / exists
  %

  For any set term $\aset$, define 
  \begin{equation}
    \Elements{\aset} = \setbuilder{ x^\structureS }{ x \opin \aset \in \setconstclosed}\ .
    \label{eq:elementsdefn}
  \end{equation}
Let $\allsetterms$ be an arbitrary set of set
terms which includes all set terms in $\setconst$. 
Using the assumption that the given state is saturated,
we show by structural induction on set terms that for any set term $\aset \in \allsetterms$:
  \begin{equation}
    \Elements{\aset} = \aset^\structureS
    \label{eq:elementsaltdefn}
  \end{equation}

  \begin{case}[$\aset$ is a variable]
    The definition of $\Elements{\aset}$ is identical to that of
    $s^\structureS$.
  \end{case}

  \begin{case}[$\aset$ is $\opemptyset$]
    Rule \ruleEmptyUnsat\ would apply to the state
    if there was a constraint of the form $x \opin \opemptyset$ in
    $\setconstclosed$.
    It follows that $\Elements{\opemptyset} = \emptyset$.
  \end{case}

  \begin{case}[$\aset$ is $\opsingleton{x}$]
    As $s^\structureS = \singleton{x^\structureS}$, it is sufficient to show that
    $\Elements{\aset} = \singleton{x^\structureS}$.
    Since rule \ruleSingle\  is not applicable, we can conclude that
    $x \opin \aset \in \setconstclosed$. It follows that
    $\singleton{x^\structureS} \subseteq \Elements{\aset}$.
    The other direction, $\Elements{\aset} \subseteq
    \singleton{x^\structureS}$, follows because of saturation with respect 
    to rule \ruleSingleMem:
    \begin{align*}
      & e \in \Elements{\opsingleton{x}}
      &
      \\
      & e = y^\structureS \text{ for some } y \text{ with } y \opin \opsingleton{x} \in \setconstclosed
      & \text{(definition)}
      \\
      & y \opequal x \in \memconstclosed
      & \text{Rule \ruleSingleMem}
      \\
      & y^\structureS = x^\structureS
      & \text{($\structureS$ satisfies $\memconst$)}
      \\
      & e \in \singleton{x^\structureS}
      & \text{($e = y^\structureS$)}
    \end{align*}
  \end{case}

  \begin{case}[$\aset$ is $\bset \opinter \cset$]
    We need to show $\Elements{\bset \opinter \cset} = \bset^\structureS \cap \cset^\structureS$.
    The proof of the left-to-right inclusion depends on rule \ruleInterDownI:
    \begin{align*}
      & e \in \Elements{\bset \opinter \cset}
      &
      \\
      & e = x^\structureS \text{ for some } x \text{ with } x \opin \bset \opinter \cset \in \setconstclosed
      & \text{(definition)}
      \\
      & x \opin \bset \in \setconstclosed \text{ and } x \opin \cset \in \setconstclosed
      & \text{(Rule \ruleInterDownI)}
      \\
      & x^\structureS \in \Elements{\bset} \text{ and } x^\structureS \in \Elements{\cset}
      & \text{(definition)}
      \\
      & x^\structureS \in \bset^\structureS \text{ and } x^\structureS \in \cset^\structureS
      & \text{(induction)}
      \\
      & e \in \bset^\structureS \cap \cset^\structureS
    \end{align*}    
    For the other direction, $\bset^\structureS \cap \cset^\structureS \subseteq \Elements{\bset \opinter \cset}$, we rely on rule \ruleInterUpI:
    \begin{align*}
      & e \in \bset^\structureS \cap \cset^\structureS
      &
      \\
      & e \in \bset^\structureS \text{ and } e \in \cset^\structureS
      &
      \\
      & e \in \Elements{\bset} \text{ and } e \in \Elements{\cset}
      & \text{(induction)}
      \\
      & x \opin \bset \in \setconstclosed \text{ and } y \opin \cset \in \setconstclosed \text{ with } x^\interpretation = y^\interpretation = e
      & \text{(definition)}
      \\
      & x \opin \bset \in \setconstclosed \text{ and } y \opin \cset \in \setconstclosed \text{ with } x \opequal y \in \memconstclosed
      & \text{(by construction)}
      \\
      & x \opin \bset \in \setconstclosed \text{ and } x \opin \cset \in \setconstclosed
      &
      \\
      & x \opin \bset \opinter \cset \in \setconstclosed
      & \text{($\bset \opinter \cset \in \allsetterms$, Rule \ruleInterUpI)}
      \\
      & e \in \Elements{t \opinter u}
      &
    \end{align*}
  \end{case}

  \begin{case}[$\aset$ is $\bset \opunion \cset$]
    First we show that $\Elements{\bset \opunion \cset}
    \subseteq \bset^\structureS \cap \cset^\structureS$:
    \begin{align*}
      & e \in \Elements{\bset \opunion \cset}
      &
      \\
      & e = x^\structureS \text{ for some } x \text{ with } x \opin \bset \opunion \cset \in \setconstclosed
      & \text{(definition)}
      \\
      %% \intertext{We show by contradiction that either $x \opin \bset \in \setconstclosed$ or $x \opin \cset \in \setconstclosed$.
      %%   %
      %%   Say, $x \opin \bset \not\in \setconstclosed$ and $x \opin \cset \not\in \setconstclosed$.
      %%   If $x \opnotin \bset \in \setconstclosed$, then because of \refrule{rule:uniondown2},
      %%   $x \opin \cset$ would be in $\setconstclosed$ leading to a contradiction.
      %%   Thus, $x \opnotin \bset \not\in \setconstclosed$.
      %%   But now, the premise of \refrule{rule:unionsplit} is satisfied,
      %%   again giving us a contradiction. It follows that,}
      & x \opin \bset \in \setconstclosed \text{ or } x \opin \cset \in \setconstclosed
      & \text{(Rule \ruleUnionUpSplit)}
      \\
      & x^\structureS \in \Elements{\bset} \text{ or } x^\structureS \in \Elements{\cset}
      & \text{(definition)}
      \\
      & x^\structureS \in \bset^\structureS \text{ or } x^\structureS \in \cset^\structureS
      & \text{(induction)}
      \\
      & e \in \bset^\structureS \cup \cset^\structureS
      &
    \end{align*}
    Then we show that $\bset^\structureS \cup \cset^\structureS \subseteq \Elements{\bset \opunion \cset}$:
    \begin{align*}
      & e \in \bset^\structureS \cup \cset^\structureS
      &
      \\
      & e \in \bset^\structureS \text{ or } e \in \cset^\structureS
      &
      \\
      & e \in \Elements{\bset} \text{ or } e \in \Elements{\cset}
      & \text{(induction)}
      \\
      & x \opin \bset \in \setconstclosed \text{ or } x \opin \cset \in \setconstclosed \text{ where } x^\interpretation = e
      & \text{(definition)}
      \\
      & x \opin \bset \opunion \cset \in \setconstclosed
      & \text{($\bset \opunion \cset \in \allsetterms$, Rule \ruleUnionUpII)}
      \\
      & e \in \Elements{t \opinter u}
      &
    \end{align*}
  \end{case}
  \begin{case}[$\aset$ is $\bset \opsetminus \cset$]
    First we show that $\Elements{\bset \opsetminus \cset} \subseteq \bset^\structureS \setminus \cset^\structureS$:
     \begin{align*}
      & e \in \Elements{\bset \opsetminus \cset}
      &
      \\
      & e = x^\structureS \text{ for some } x \text{ with } x \opin \bset \opsetminus \cset \in \setconstclosed
      & \text{(definition)}
      \\
      & x \opin \bset \in \setconstclosed \text{ and } x \opnotin \cset \in \setconstclosed
      & \text{(Rule \ruleSetDifferenceDownI)}
      \\
      \intertext{From $x \opnotin \cset \in \setconstclosed$ we can conclude 
      that $x^\structureS \in \Elements{\cset}$.
      In fact, if we assume otherwise, we have that 
      $y \opin \cset \in \setconstclosed$ for some $y$ with $x^\structureS = y^\structureS$.
      But then $x \opequal y \in \memconstclosed$ which implies that 
      $x \opin \cset \in \setconstclosed$.
      This, however, make rules \ruleSetUnsat\ applicable. Then we have:}
      & x^\structureS \in \Elements{\bset} \text{ and } x^\structureS \not\in \Elements{\cset}
      &
      \\
      & x^\structureS \in \bset^\structureS \text{ and } x^\structureS \not\in \cset^\structureS
      & \text{(induction)}
      \\
      & e \in \bset^\structureS \setminus \cset^\structureS
      &
    \end{align*}
    We now show that $\bset^\structureS \setminus \cset^\structureS \subseteq \Elements{\bset \opsetminus \cset}$:
    \begin{align*}
      & e \in \bset^\structureS \setminus \cset^\structureS
      &
      \\
      & e \in \bset^\structureS \text{ and } e \not\in \cset^\structureS
      &
      \\
      & e \in \Elements{\bset} \text{ and } e \not\in \Elements{\cset}
      & \text{(induction)}
      \\
      & x \opin \bset \in \setconstclosed \text{ and } x \opin \cset \not\in \setconstclosed \text{ for some } x \text{ with } x^\interpretation = e
      & \text{(definition)}
    \end{align*}
      We show by contradiction that $x \opnotin \cset \in \setconstclosed$. 
      Assume the otherwise. Since $x \opin \cset \not\in \setconstclosed$ and 
      $\bset \opsetminus \cset \in \allsetterms$, the premise of rule \ruleSetDifferenceSplit~is satisfied. 
      As we had neither $x \opin \cset \in \setconstclosed$ nor $x \opnotin \cset \in \setconstclosed$, we get a contradiction.
    \begin{align*}
      & x \opin \bset \in \setconstclosed \text{ and } x \opnotin \cset \in \setconstclosed
      & \text{(Rule \ruleSetDifferenceSplit)}
      \\
      & x \opin \bset \opsetminus \cset \in \setconstclosed
      & \text{($\bset \opsetminus \cset \in \allsetterms$, Rule \ruleSetDifferenceUpI)}
      \\
      & e \in \Elements{t \opsetminus u}
      &
    \end{align*}
  \end{case}
  
  Having established the property of $\Elements{\cdot}$, showing that each
  constraint in $\setconst$ is satisfied by $\structureS$ is
  straightforward:
  \begin{enumerate}
    \item  
      Let $x \opin \aset \in \setconst$. Then, $x^\structureS \in
      \Elements{\aset}$ by (\ref{eq:elementsdefn})
 and $x^\structureS \in \aset^\structureS$ by (\ref{eq:elementsaltdefn}).
    \item 
      Let $x \opnotin \aset \in \setconst$. We show $x^\structureS
      \not\in \aset^\structureS$ by contradiction.
    \begin{align*}
      & x^\structureS \in \aset^\structureS
      & \text{(assume)}
      \\
      & x^\structureS \in \Elements{\aset}
      & \text{(proved above)}
      \\
      & x^\structureS = y^\structureS \text{ for some } y \text{ with } y \opin \aset \in \setconstclosed
      & \text{(definition)}
      \\
      & x \opequal y \in \memconstclosed
      & \text{($x^{\structureS} = y^\structureS$ iff $x \opequal y \in \memconstclosed$)}
      \\
      & x \opin \aset \in \setconstclosed
      & \text{(definition of $\setconstclosed$)}
      \\
      & \text{Tableau is closed, contradiction.}
      & \text{(Rule \ruleSetUnsat)}
    \end{align*}
    \item 
      Let $\aset \opequal \bset \in \setconst$.
      From the definition of $\setconstclosed$
      it follows that $\Elements{\aset} = \Elements{\bset}$.
      Since $s^\structureS = \Elements{s}$ and $t^\structureS = \Elements{t}$,
      it follows that $s^\structureS = t^\structureS$.
    \item
      Let $s \not\opequal t \in \setconst$. From rule \ruleSetDiseq, it follows that there exists
      $x$ such that either $x \opin \aset \in \setconstclosed$ and $x \opnotin \bset \in \setconstclosed$, or
      $x \opnotin \aset \in \setconstclosed$ and $x \opin \bset \in \setconstclosed$.
      It follows that either $x^\structureS \in \aset^\structureS$ and $x^\structureS \not\in \bset^\structureS$, or
      $x^\structureS \not\in \aset^\structureS$ and $x^\structureS \in \bset^\structureS$.
      In either case, we can conclude that $\aset^\structureS \neq \bset^\structureS$.
    \item Let $c_\asetv \opequal \opcard{\asetv} \in \setconst$. By definition, both $c_\asetv^\structureS = \card{\asetv^\structureS} = \opcard{\asetv}^\structureS$. \qedhere
%      \csays{Somewhere you need to say that $\setconst$ does not contains cardinality
%      constraints.}
  \end{enumerate}
  
  %% We do a case analysis on each kind of constraints possible in $\setconst$.
  %% \begin{enumerate}
  %% \item $S \opequal T$. Follows from definition, since we are looking
  %%   for member constraints in $\setconstclosed$.

  %%   $S \not\opequal T$. Disequality rules.
  %% \item $S \opequal T \opunion U$. Need to show $S^\structureS = T^\structureS \cup U^\structureS$.

  %%   Let $e \in S^\structureS$. Then there must exists an $x$ with
  %%   $x^\structureS = e$ such that $x \opin S \in \setconstclosed$.
  %%   %
  %%   As $S \opequal T \opunion U$, $x \opin T \opunion U \in \setconstclosed$.
  %%   %
  %%   Because of rules for union, we have that either $x \opin T \in
  %%   \setconstclosed$ or $x \opin U \in \setconstclosed$.

  %%   Case 1: 
    
  %%   Thus, either $x^\structureS \in T^\structureS$ or $x^\structureS \in
  %%   U^\structureS$. It follows that $e \in T^\structureS \cup
  %%   U^\structureS$.

  %%   To show the other direction, let $e \in T^\structureS \cup U^\structureS$.

  %%   $\ldots$
  %% \end{enumerate}

  % x \in S
  % y \in T
  % S \inter T
  % (x ~ y) => [assume things to be disequal]
\end{proof}

%%% Local Variables:
%%% mode: latex
%%% TeX-master: "pap.tex"
%%% End:

%\subsubsection{Properties of graph constraints}
For the next two results, let $\langle \setconst, \memconst, \arithconst, \graphconst \rangle$ 
be a derivation tree leaf saturated with respect to rules $\mathcal{R}_1 \cup \mathcal{R}_2 \cup \mathcal{R}_3$ 
in a derivation tree.
The first result is about the effects of the rules in $\mathcal{R}_2$.
The second is about the rules in $\mathcal{R}_3$.

\begin{prop} \label{prop:graph}
%CT Streamlined the statement of the proposition
%
For every $\aset \in \vertices$ the following holds.
\begin{enumerate}
\item \label{prop:graph:propequal}
  If $\aset \opequal \bset \in \setconst$ or $\bset \opequal \aset \in \setconst$
  for some $\bset$, then $\nonemptyleaves{\aset} = \nonemptyleaves{\bset}$.
\item \label{prop:graph:propunion}
  If $s = T \opunion U$, then $\nonemptyleaves{T \opunion U} = \nonemptyleaves{T} \cup \nonemptyleaves{U}$.
\item \label{prop:graph:propinter}
  If $s = T \opinter U$, then $\nonemptyleaves{T \opinter U} = \nonemptyleaves{T} \cap \nonemptyleaves{U}$.
\item \label{prop:graph:propsetminus}
  If $s = T \setminus U$, then $\nonemptyleaves{T \opsetminus U} = \nonemptyleaves{T} \setminus \nonemptyleaves{U}$.
\item \label{prop:graph:propdisjoint}
  For all distinct $\bset, \cset \in \leaves{\aset}$, %the two sets are necessarily disjoint:
  $\models_\settheory  \bset \opinter \cset \opequal \opemptyset$.

\item \label{prop:graph:propleaves}
  \(
  \setbuilder{\bset \opequal \cset}{\bset \opequal \cset \in \setconstclosed}
  \models_\settheory
  \aset \opequal \bigsqcup_{\bset \in \nonemptyleaves{\aset}} \bset .
  \)%
  \footnote{Technically, $\bigsqcup_{\ldots}$ is ambiguous. However, since $\opunion$ is associative in $\settheory$, bracketing does not matter in this context.}
\end{enumerate}
\end{prop}
\begin{proof}[Proof (Proposition \ref{prop:graph}, property \ref{prop:graph:propequal})]
Let $\aset \opequal \bset \in \setconst$, with $\aset \in \vertices$ or $\bset \in \vertices$.
From rule \ruleIntroEqRight~and rule \ruleIntroEqLeft\ it follows that both $\aset \in \vertices$ and $\bset \in \vertices$.
For each of the Rules \ruleMergeEqI, \ruleMergeEqII, and \ruleMergeEqIII; we show that
after the application of the rule, $\nonemptyleaves{s}$ and $\nonemptyleaves{t}$ are equal.
%\begin{itemize}
%\item

Consider rule \ruleMergeEqI.
Let $L_\aset$ and $L_\bset$ denote $\nonemptyleaves{\aset}$ and $\nonemptyleaves{\bset}$ respectively
before application of the rule. Let $L'_\aset$ and $L'_\bset$ denote
$\nonemptyleaves{\aset}$ and $\nonemptyleaves{\bset}$ after application of the rule.
For the rule to be applicable $L_\aset \subsetneq L_\bset$.
The rule adds constraints to $\setconst$ so that $L'_\bset = L_\bset \setminus (L_{\bset} \setminus L_\aset)$.
Equivalently, $L'_\bset = L_\bset \cap L_\aset = L_\aset$. Since $L'_\aset = L_\aset$,
we get $L'_\aset = L_\aset = L'_\bset$.

%\item
The case for rule \ruleMergeEqII~is analogous to rule \ruleMergeEqI.

%\item
Consider rule \ruleMergeEqIII.
Let $L_\aset$ and $L_\bset$ denote $\nonemptyleaves{\aset}$ and $\nonemptyleaves{\bset}$ respectively
before application of the rule. Let $L'_\aset$ and $L'_\bset$ denote
$\nonemptyleaves{\aset}$ and $\nonemptyleaves{\bset}$ after application of the rule.
Let $n \in L'_\aset$. Note that the $\merge$ operation only adds nodes
and vertices. Thus, $n$ is one of the following:
\begin{itemize}
\item $l_1 \opinter l_2$ with $l_1 \in L_\aset$ and $l_2 \in L_\bset$:
  Since $(l_1, l_1 \opinter l_2)$ as well as $(l_2, l_1 \opinter l_2)$ is an edge, it follows that $n \in L'_\bset$.
\item $l_1 \in L_\aset$. Since nodes in $L_\aset \setminus L_\bset$ have an outgoing edge, it must be the case that $l_1 \in L_\aset \cap L_\bset$. It follows that $n \in L'_\bset$.
\end{itemize}
This shows that $L'_\aset \subseteq L'_\bset$. The reasoning for $L'_\bset \subseteq L'_\aset$ is symmetrical.
%\end{itemize}

As $\aset \opequal \bset$, $\aset \in \vertices$, and $\bset \in \vertices$, the premise of at least
once of the rules (\ruleMergeEqI), (\ruleMergeEqII), and (\ruleMergeEqIII) must be satisfied
whenever $\nonemptyleaves{s} \neq \nonemptyleaves{t}$.
As the branch is saturated, $\nonemptyleaves{s} = \nonemptyleaves{t}$ follows.
\end{proof}
%\item
%
\begin{proof}[Proof (Proposition \ref{prop:graph}, properties \ref{prop:graph:propunion}, \ref{prop:graph:propinter}, \ref{prop:graph:propsetminus})]

As $\derivtree$ is obtained from a derivation starting with a state
with an empty graph, it is sufficient to show the properties hold for
the empty graph, and that they are preserved each time the graph is
modified by one of the rules.

The properties hold trivially for the empty graph. The interesting
cases are when edges are added to the graph: i) $\add$ of a union,
intersection, or set minus term, and ii) $\merge$ operation.

\reviewer{
 - page 18: line 14 from the bottom:  sounds like there is something missing in
 the sentence "Observe that when we..."
}
\response{
  Added ``the following holds'' to end of sentence.
}
Observe that when we introduce $T \opunion U$,  $T \opinter U$, and  $T \opsetminus U$ to the graph, the following holds:
\begin{itemize}
\item $\leaves{T} = \{T \opsetminus U, T \opinter U\}$,
\item $\leaves{U}= \{T \opinter U, U \opsetminus T\}$,
\item  $\leaves{T \opunion U} = \{ T \opsetminus U,  T \opinter U, U \opsetminus T\}$,
\item  $\leaves{T \opinter U} = \{ T \opinter U \}$,
\item $\leaves{T \opsetminus U} = \{T \opsetminus U\}$, and
\item $\leaves{U \opsetminus T} = \{U \opsetminus T\}$.
\end{itemize}
We conclude that:
\begin{itemize}
\item $\leaves{T \opunion U} = \leaves{T} \cup \leaves{U}$
\item $\leaves{T \opinter U} = \leaves{T} \cap \leaves{U}$
\item $\leaves{T \opsetminus U} = \leaves{T} \setminus \leaves{U}$
\item $\leaves{U \opsetminus T} = \leaves{U} \setminus \leaves{T}$
\end{itemize}
when an introduce rule is applied.
Note that the merge operation only adds edges from existing leaf
nodes, ensuring that the property is maintained by any application
of $\merge$.

$\nonemptyleaves{\cdot}$, as defined in (\ref{eq:nonemptyleaves}), can also be defined as:
\begin{equation}
  \nonemptyleaves{n} = \leaves{n} \setminus E
\end{equation}
where $E = \setbuilder{n' \in \vertices}{n' \opequal \opemptyset \in \setconstclosed}$
does not depend on $n$.
The properties in the proposition about $\nonemptyleaves{\cdot}$
follow from the corresponding property of $\leaves{\cdot}$ just established, and
above formulation of $\nonemptyleaves{\cdot}$.
\end{proof}
%% \item Let $\aset \in \vertices$.
%%   For all $\bset, \cset \in \leaves{\aset}$, $\bset \neq \cset$, the two sets are necessarily disjoint:
%%   \[ \models_\settheory  \bset \opinter \cset \opequal \opemptyset\text{.} \]
\begin{proof}[Proof (Proposition \ref{prop:graph}, properties \ref{prop:graph:propdisjoint},\ref{prop:graph:propleaves})]
The properties holds trivially for the empty graph.

Let $\graphconst$ be the graph constraints. Let $\aset \in \vertices$.
Let $\aset' \opequal \emptyset$ be a new constraint such that $\aset' \in \nonemptyleaves{\aset}$.
Then, this modifies $\nonemptyleaves{\aset}$, and we need to verify the Property \ref{prop:graph:propleaves} still holds.
%% Property \ref{prop:graph:propdisjoint} is not affected,
%% the interesting case is property \ref{prop:graph:propleaves}.
%% If
%% \[\models_\settheory \inparen{\bigwedge_{\bset \in E} \bset \opequal \opemptyset} \Rightarrow \aformula \quad\text{,}\]
%% then
%% \[\inparen{\aset' \opequal \opemptyset \wedge \bigwedge_{\bset \in E} \bset \opequal \opemptyset} \Rightarrow \aformula \quad\text{.}\]
Note that for any structure in $\settheory$, if $\aset'$ is interpreted as empty set, the interpretation of
$\bigsqcup_{\bset \in \nonemptyleaves{\aset} \setminus \{\aset'\}} \bset$
will be same as $\bigsqcup_{\bset \in \nonemptyleaves{\aset}} \bset$.
Thus, if $\aset' \in \nonemptyleaves{\aset}$ and
\[ \models_\settheory \inparen{\bigwedge_{P \in E} P} \Rightarrow \inparen{\aset \opequal \bigsqcup_{\bset \in \nonemptyleaves{\aset}} \bset} \quad\text{,}\]
then
 \[ \models_\settheory \inparen{\aset' \opequal \opemptyset \wedge \bigwedge_{P \in E} P} \Rightarrow \inparen{\aset \opequal \bigsqcup_{\bset \in \nonemptyleaves{\aset} \setminus \{\aset'\}} \bset} \quad\text{.}\]
It follows if $\aset' \opequal \emptyset$ is added to $\setconstclosed$ by a rule, the property \ref{prop:graph:propleaves} continue to hold.
Also note that an equality is not removed by any rule (if there was such a rule, we would need to check the property continues to hold when the left side of the implication is weakened).

The only other rules which affect the properties are those which modify the graph directly,
i.e. the $\add$ and $\merge$ operations.

We show that if $\graphconst$ satisfies the properties, then so does $\add(\graphconst, \aset)$:
\begin{itemize}
\item $\aset$ is $\opemptyset$, $\asetv$ or $\opsingleton{x}$: trivially, as no edges are added.
\item $\aset$ is $\bsetv \opinter \csetv$:
  Note that because of the assumptions on the normal form, either $\bsetv \opinter \csetv$ already in the graph and $\add$
  operation does not modify the graph, or it will add the nodes $\bsetv$, $\csetv$, $\bsetv \opsetminus \csetv$, $\bsetv \opinter \csetv$,
  and $\csetv \opsetminus \bsetv$ to the graph, and edges between them.
%\ksays{The assumption does not give this unfortunately. Normal form still needs
%to be adjusted. Merge can add edges.}
 It is easy to see that the property \ref{prop:graph:propdisjoint} follows from:
  \begin{align*}
    &\models_\settheory \inparen{\inparen{\bsetv \opsetminus \csetv} \opinter \inparen{\bsetv \opinter \csetv}} \opequal \opemptyset \\
    &\models_\settheory \inparen{\inparen{\csetv \opsetminus \bsetv} \opinter \inparen{\bsetv \opinter \csetv}} \opequal \opemptyset
  \end{align*}
  Property \ref{prop:graph:propleaves} follows from:
  \begin{align*}
    &\models_\settheory \bsetv \opequal \inparen{\inparen{\bsetv \opsetminus \csetv} \opunion \inparen{\bsetv \opinter \csetv}} \\
    &\models_\settheory \csetv \opequal \inparen{\inparen{\csetv \opsetminus \bsetv} \opunion \inparen{\bsetv \opinter \csetv}} \\
    &\models_\settheory \inparen{\bsetv \opinter \csetv} \opequal \inparen{\bsetv \opinter \csetv} \\
    &\models_\settheory \inparen{\csetv \opsetminus \bsetv} \opequal \inparen{\csetv \opsetminus \bsetv} \\
    &\models_\settheory \inparen{\bsetv \opsetminus \csetv} \opequal \inparen{\bsetv \opsetminus \csetv}
  \end{align*}
  and reasoning as earlier that any constraint of the form $\aset' \opequal \opemptyset$ does not affect the property.
\item $\aset$ is $\bsetv \opsetminus \csetv$ or $\csetv \opsetminus \bsetv$: reasoning same as for $\bsetv \opinter \csetv$.
\item $\aset$ is $\bsetv \opunion \csetv$. If not already present, $\bsetv$, $\csetv$, $\bsetv \opsetminus \csetv$, $\bsetv \opinter \csetv$ are added to the graph as for $\bsetv \opinter \csetv$. In addition, $\add$ for union also adds $\bsetv \opunion \csetv$, and three edges. The properties follows from the following tautologies in $\settheory$ in addition to those listed in analysis for $\bsetv \opinter \csetv$:
  \begin{align*}
    &\models_\settheory \inparen{\inparen{\bsetv \opsetminus \csetv} \opinter \inparen{\inparen{\csetv \opsetminus \bsetv}}} \opequal \opemptyset \\
    &\models_\settheory \inparen{\bsetv \opunion \csetv} \opequal \inparen{\inparen{\bsetv \opsetminus \csetv} \opunion \inparen{\bsetv \opinter \csetv} \opunion \inparen{\csetv \opsetminus \bsetv}}
  \end{align*}
\end{itemize}

Finally, we show that if $\graphconst$ satisfies the properties,
then so does $\merge(\graphconst, \aset, \bset)$ if
$\aset \in \vertices$,
$\bset \in \vertices$,
$\nonemptyleaves{\aset} \nsubseteq \nonemptyleaves{\bset}$ and
$\nonemptyleaves{\bset} \nsubseteq \nonemptyleaves{\aset}$.

Let $L_\aset$ denote $\nonemptyleaves{\aset}$ in $\graphconst$,
and $L'_\aset$ denote $\nonemptyleaves{\aset}$ in $\merge(\graphconst, \aset', \bset')$
(likewise for $\bset$, $\cset$ etc.).

In order to show property \ref{prop:graph:propdisjoint} holds,
let $\aset' \in \vertices$, $\bset' \in L'_{\aset'}$ and $\cset' \in L'_{\aset'}$.
We need to show:
$\models_\settheory \bset' \opinter \cset' \opequal \opemptyset$.
\begin{itemize}
\item Let $\bset' \in L_{\aset'}$ and $\cset' \in L_{\aset'}$, i.e. both are
  also leaf nodes in $\graphconst$. Then, the property for $\merge(\graphconst, \aset, \bset)$
  follows from that of $\graphconst$.
\item Let $\bset'$ be one of the newly introduced leaf nodes and $\cset' \in L_{\aset'}$ a leaf node in $\graphconst$.
  Without loss of generality, let $\bset'$ be $t_1 \opinter t_2$
  with $t_1 \in L_\aset \setminus L_\bset$ and $t_2 \in L_\bset \setminus L_\aset$.
  For $\bset'$ to be in $L'_{\aset'}$, given the way the edges are added, either
  $t_1 \in L_{\aset'}$ or $t_2 \in L_{\aset'}$. Thus, we know that either
  $\models_\settheory t_1 \opinter \cset' \opequal \opemptyset$
  or
  $\models_\settheory t_2 \opinter \cset' \opequal \opemptyset$.
  In either case, it follows that
  $\models_\settheory \inparen{t_1 \opinter t_2} \opinter \cset' \opequal \opemptyset$, i.e.
  $\models_\settheory \bset' \opinter \cset' \opequal \opemptyset$.
\item The analysis for the case where both are newly introduced leaf nodes is similar.
\end{itemize}

To show property \ref{prop:graph:propleaves} holds,
the main observation is that each node no longer a leaf node, say $\aset' \in L_{\aset} \setminus L'_{\aset}$,
is union of a new set of leaf nodes in $L'_{\aset}$ (assuming the equalities).
\begin{align*}
  \aset'
  & \opequal \aset' \opinter \aset
  & \text{($\aset' \in L_\aset$, $\aset \opequal \bigsqcup_{\aset'' \in L_\aset} \aset''$)}
  \\
  & \opequal \aset' \opinter \bset
  & \text{($\aset \opequal \bset \in E$)}
  \\
  & \opequal \aset' \opinter \inparen{\bigsqcup_{\bset' \in L_\bset} \bset'}
  & \text{($\bset \opequal \bigsqcup_{\bset' \in L_\bset} \bset'$)}
  \\
  & \opequal \bigsqcup_{\bset' \in L_\bset} \aset' \opinter \bset'
  & \text{(distribute)}
  \\
  \intertext{But by property \ref{prop:graph:propdisjoint}, $\aset' \opinter \bset' \opequal \opemptyset$ for $\aset', \bset' \in L_\aset$. Thus,}
  \aset'
  & \opequal \bigsqcup_{\bset' \in L_\bset \setminus L_\aset} \aset' \opinter \bset'
\end{align*}
Note that $\setbuilder{\aset' \opinter \bset'}{\bset'\in L_\bset \setminus L_\aset}$ are precisely the nodes in
$L'_\aset$ to which edges are added from $\aset'$. The proof for a node in $L_{\bset}$ but not in $L'_{\bset}$ is similar.

Since all the new leaf nodes are of the form $\aset' \opinter \bset'$ with $\aset'\in L_\aset \setminus L_\bset$ and
$\bset'\in L_\bset \setminus L_\aset$, it follows that property \ref{prop:graph:propleaves} holds for $\merge(\graphconst, \aset, \bset)$ if
it holds for $\graphconst$ assuming $\aset \opequal \bset \in E$.
%
%% \item Let $\aset \in \vertices$. Let $E = \setbuilder{\bset \in \vertices}{\bset \opequal \opemptyset \in \setconstclosed}$.
%%   Then\footnote{Technically, $\bigsqcup_{\ldots}$ is ambiguous. But, in any interpretation in $\settheory$, the interpretation of $\opunion$ is associative, so the bracketing does not matter in our context.}
%%   \begin{equation}
%%   \models_\settheory \inparen{\bigwedge_{\bset \in E} \bset \opequal \opemptyset}
%%   \Rightarrow
%%   \inparen{\aset \opequal \bigsqcup_{\bset \in \nonemptyleaves{\aset}} \bset}
%%   \label{eq:nodeLeavesReln}
%%   \end{equation}
\end{proof}

%%% Local Variables:
%%% mode: latex
%%% TeX-master: "pap.tex"
%%% End:

%% \begin{prop} \label{prop:arithmodel}
%% Let $\structureS$ be an interpretation as the one specified in Proposition \ref{prop:membermodel} and
%% let $\structureA$ be any model of $\settheory$ satisfying $\arithconst$.
%% Then, for all $\bset \in \nonemptyleaves{\graphconst}$,
%% \(
%%   c_\bset^\structureA \geq \card{\bset^\structureS}.
%% \)
%% \end{prop}%
\begin{prop} \label{prop:arithmodel}
Let
$\langle \setconst, \memconst, \arithconst, \graphconst \rangle$
be a state such that none of the rules in our calculus are applicable.
Let $\structureS$ be an  interpretation defined in Proposition \ref{prop:membermodel}
satisfying constraints in $\setconst$ and $\memconst$.
To recall, for $x$ and $y$ of $\elemSort$ sort,
\[ x^{\structureS} = y^\structureS  \text{ if and only if } x \opequal y \in \memconstclosed \]
and for $\aset$ of $\setSort$ sort,
\[
  \aset^\structureS = \setbuilder{ x^\structureS }{ x \opin s \in \setconstclosed}\text{.}
\]
Let $\structureA$ be an interpretation satisfying $\arithconst$.
Then, for all $\bset \in \nonemptyleaves{\graphconst}$,
\[
  c_\bset^\structureA \geq \card{\bset^\structureS}\text{.}
\]
\end{prop}%
\begin{proof}
  Let $\bset \in \nonemptyleaves{\graphconst}$.
  First we show that if
  $\arithconst \Rightarrow \cardbset \geq \card{\bset_\setconst}$, then the proposition follows.
  That is there exists $n \geq \card{\bset_\setconst}$
  such that $\cardbset \opgeq n \in \arithconst$.
  Let $\Elements{\cdot}$ be as in (\ref{eq:elementsdefn}).
  \begin{align*}
    c_\bset^\structureA
    & \geq n^\structureA
    & \text{($\cardbset \opgeq n \in \arithconst$)}
    \\
    & = n
    & \text{(constant symbol)}
    \\
    & \geq \card{\bset_\setconst}
    & \text{(definition)}
    \\
    & = \card{\Elements{\bset}}
    & \text{($x^{\structureS} = y^\structureS$ iff $x \opequal y \in \memconstclosed$)}
    \\
    & = \card{\bset^\structureS}
    & \text{(using (\ref{eq:elementsaltdefn}))}
  \end{align*}
  It remains to show that $\arithconst \Rightarrow \cardbset \geq \card{\bset_\setconst}$.
  Because of rule \ruleMemArrange, either
  $\arithconst \Rightarrow \cardbset \geq \card{\bset_\setconst}$
  or
  Rule \ruleMemArrange\ is applicable until the premise of
  rule \rulePropMinsize\ holds.
  If Rule \rulePropMinsize\ is applicable,
  $\cardbset \opgeq \card{\bset_\setconst}$ must have been added to $\arithconst$.
  In either case, $\arithconst \Rightarrow \cardbset \geq \card{\bset_\setconst}$.
\end{proof}

\noindent
Completeness is a direct consequence of the following result.

  \newcommand{\setsStatementOfCompletenessProposition}{%
    Let $\setconst_0, \memconst_0, \arithconst_0$ be
    set, element and cardinality constraints respectively, satisfying Restriction
    \ref{restr}.
    Let $\derivation$ be a derivation with respect to rules
    $\mathcal{R}_1 \cup \mathcal{R}_2 \cup \mathcal{R}_3$
    from state
    $\langle \setconst_0,$ $\memconst_0,$ $\arithconst_0,$ $(\emptyset, \emptyset) \rangle$.
    If $\derivation$ is finite, and the final derivation tree, say $\derivtree$,
    in $\derivation$ is open and saturated with respect to the rules
    $\mathcal{R}_1 \cup \mathcal{R}_2 \cup \mathcal{R}_3$;
    then there exists an interpretation
    $\structure$ that satisfies $\setconst_0$, $\memconst_0$ and
    $\arithconst_0$.}%

\begin{prop}\label{prop:solution-soundness}
\setsStatementOfCompletenessProposition
\end{prop}
\begin{proof}%[Proof of Proposition \ref{prop:completeness}]
Proof outline:
We build a model of the leaf nodes in the graph by modifying as needed the model obtained 
from Proposition~\ref{prop:membermodel}.
We add additional elements to these sets to make the cardinalities match
the model satisfying the cardinality constraints and the constraints induced by the graph.
Propositions \ref{prop:graph} and \ref{prop:arithmodel} ensure that 
it is always possible to do so without violating the set constraints.

As $\derivtree$ is open, there exists a branch that does not
end in the state \textsf{unsat}. Let
$\langle \setconst, \memconst, \arithconst, \graphconst \rangle$
be the final state on such a branch.

Let $\arithconst \cup \grapharithconst$ be the cardinality constraints, and
the cardinality constraints induced by the graph. These constraints fall
in the theory $\ariththeory$. Let $\structureA$ be the
structure satisfying these constraints. Such a structure exists
because rule \ruleAriContra\ would have closed the branch
if the constraints were inconsistent.
From Proposition \ref{prop:membermodel}, we obtain a structure
$\structureS$ satisfying $\setconst$ and $\memconst$. Without loss of
generality, assume that $\elemSort^\structureS$ is infinite.

The $\structure$ we build satisfying $\setconst_0 \cup \memconst_0
\cup \arithconst_0$ will be as follows. It coincides with the
structure $\structureS$ on terms of $\elemSort$ sort. It coincides
with the structure $\structureA$ on terms of $\cardSort$ sort.
In order to define the value of set variables, for each leaf
node $\bset \in \allleaves$ we create the following sets:
%\[ A_\bset = \bset^\structureS = \setbuilder{ x^\structure }{ x \opin \bset \in \setconstclosed } \]
\[ B_\bset = \{ e_{\bset,1}, e_{\bset,2} \ldots e_{\bset,\cardbset^\structure-\card{\bset^\structureS}} \} \]
where $e_{\bset,i} \in \elemSort^\structureS$ are distinct from each other
and from any $e$ such that $e=x^\structureS$ for $x$ in $\setconst$ or $\memconst$.
From Proposition \ref{prop:arithmodel}, we know that $\cardbset^\structure \geq \card{\bset^\structureS}$.
Thus, for a leaf node $t$,
\begin{equation}
  \card{\bset^\structureS} + \card{B_t} = \cardbset^\structure\text{.}
  \label{eq:settermcard}
\end{equation}

For a set variable not in the graph, $\asetv \not\in \vertices$,
define $\asetv^\structure = \asetv^\structureS$.
For a set variable in the graph, $\asetv \in \vertices$, define:
\begin{equation}
  \asetv^\structure = \bigcup_{\bset \in \nonemptyleaves{\asetv}} (\bset^\structureS \cup B_\bset)
  \label{eq:setvarmodel}
\end{equation}
From Proposition \ref{prop:graph}, it follows that:
\begin{equation}
  \bigcup_{\bset \in \nonemptyleaves{\asetv}} \bset^\structureS = \asetv^\structureS
  \label{eq:unionOfLeavesIsNode}
\end{equation}
So an equivalent way to define $\asetv^\structure$ is as follows:
\begin{equation}
  \asetv^\structure = \asetv^\structureS \cup \bigcup_{\bset \in \nonemptyleaves{\asetv}} B_\bset
  \label{eq:setvarmodelalt}
\end{equation}

We verify that each constraint in $\setconst_0$ is satisfied:
\begin{enumerate}
\item $S \opequal T$, $S \not\opequal T$.
  
  For $S \opequal T$, we need to show $S^\structure =  T^\structure$.
  If neither $S \in \vertices$ nor $T \in \vertices$, then this
  follows from Proposition \ref{prop:membermodel}.
  If either $S \in \vertices$ or $T \in \vertices$,
  then due to rule \ruleIntroEqRight\ and rule \ruleIntroEqLeft\ 
  both $S \in \vertices$ and $T \in \vertices$.
  From \refgraphprop{equal}, we know that $\nonemptyleaves{S} = \nonemptyleaves{T}$.
  From the definition of $S^\structure$ and $T^\structure$ in (\ref{eq:setvarmodel}),
  it follows that $S^\structure = T^\structure$.

  For $S \not\opequal T$, we need to show $S^\structure \neq T^\structure$.
  Let us write $S^\structure = S^\structureS \cup B_S$, where $B_S = \emptyset$ if $S \not\in \vertices$,
  otherwise let $B_S = \bigcup_{\bset \in \nonemptyleaves{\asetv}} B_\bset$ (from (\ref{eq:setvarmodelalt})).
  Similarly we may write $T^\structure = T^\structureS \cup B_T$.
  From Proposition \ref{prop:membermodel} we know that $S^\structureS \neq T^\structureS$.
  Without loss of generality assume $e \in S^\structureS$ and $e \not\in T^\structureS$.
  By definition, $B_T$ is disjoint from $S^\structureS$, thus $e \not\in B_T$.
  Thus, $e \in S^\structure$ and $e \not\in T^\structure$.  $S^\structure \neq T^\structure$ follows.
\item $S \opequal \opemptyset$.

  We need to show $S^\structure = \opemptyset^\structure = \emptyset$.
  It will follow from rule \ruleIntroEqEmpty\ and rule \ruleIntroEqLeft.
  \begin{align*}
    & \opemptyset \in \vertices \text{ and }  S \in \vertices
    & \text{(Rules \ruleIntroEqEmpty, \ruleIntroEqLeft)}
    \\
    & \nonemptyleaves{S} = \nonemptyleaves{\opemptyset}
    & \text{(\refgraphprop{equal})}
    \\
    & \nonemptyleaves{S} = \emptyset
    & \text{($\nonemptyleaves{\opemptyset} = \emptyset$)}
    \\
    & S^\structure = \emptyset
    & \text{($S \in \vertices$, (\ref{eq:setvarmodel}))}
  \end{align*}
\item $S \opequal \opsingleton{x}$.%\ksays{Not sure this is the simplest reasoning, but think this works.}

  We need to show that $S^\structure = \singleton{x^\structure}$.
  From rule \ruleIntroEqSingle\ we conclude that $\opsingleton{x} \in \vertices$
  Then, from rule \ruleIntroEqLeft, $S \in \vertices$.

  From $\grapharithconst$, we know that:
  \begin{align*}
    c_\asetv^\structure
    & = \sum_{\bset \in \nonemptyleaves{\asetv}} \cardbset^\structure
    & \text{(constraint in $\grapharithconst$ for $c_\asetv$)}\\
    & = \sum_{\bset \in \nonemptyleaves{\opsingleton{x}}} \cardbset^\structure
    & \text{(\refgraphprop{equal})} \\
    & = c_{\opsingleton{x}}^\structure
    & \text{(constraint in $\grapharithconst$ for $c_{\opsingleton{x}}$)}\\
    & = 1
    & \text{(constraint in $\grapharithconst$ for singletons)}
  \end{align*}
  We can conclude that $\card{\asetv^\structure}=1$
  as $\card{\asetv^\structure} = c^\structure_\asetv$
  (for proof of $\card{\asetv^\structure} = c^\structure_\asetv$,
  see reasoning later in this proof for
  $\card{\asetv} \opequal c_\asetv$ -- the same reasoning works for all nodes
  $\asetv \in \vertices$)

  From, \ruleSingle, we know $x^\structureS \in S^\structureS$.
  By Proposition \ref{prop:membermodel}, $x^\structureS \in \asetv^\structureS$.
  As
  \[\asetv^\structure = \asetv^\structureS \cup \bigcup_{\bset \in \nonemptyleaves{\asetv}} B_\bset\]
  and
  $\card{\asetv^\structure} = 1$, we conclude that
  $\asetv^\structure = \singleton{x^\structureS} = \singleton{x^\structure}$.
\item $\asetv \opequal \bsetv \opunion \csetv$.
  We need to show $\asetv^\structure =  \bsetv^\structure \cup \csetv^\structure$.

  Let $\asetv \not\in \vertices$, $\bsetv \not\in \vertices$, and $\csetv \not\in \vertices$.
  Then,
  \begin{align*}
    \asetv^\structure
    & = \asetv^\structureS
    & \text{($\asetv \not\in \vertices$)}
    \\
    & = \bsetv^\structureS \cup \csetv^\structureS
    & \text{(Proposition \ref{prop:membermodel})}
    \\
    & = \bsetv^\structure \cup \csetv^\structure
    & \text{($\bsetv \not\in \vertices$, $\csetv \not\in \vertices$)}
  \end{align*}

  Otherwise, let $\asetv \in \vertices$, or $\bsetv \in \vertices$, or $\csetv \in \vertices$.
  Then, from Rules \ruleIntroEqRight, \ruleIntroEqLeft, \ruleIntroEqUnion\ and
  definition of $\add$, we know $\asetv$, $\bsetv$, \emph{and} $\csetv$ in $\vertices$.
  Then,
  \begin{align*}
    \asetv^\structure
    & = \bigcup_{\bset \in \nonemptyleaves{\asetv}} (\bset^\structureS \cup B_\bset)
    & \text{($\asetv \in \vertices$)}
    \\
    & = \bigcup_{\bset \in \nonemptyleaves{\bsetv \opunion \csetv}} (\bset^\structureS \cup B_\bset)
    & \text{(Proposition \ref{prop:graph})}
    \\
    & = \inparen{\bigcup_{\bset \in \nonemptyleaves{\bsetv}} (\bset^\structureS \cup B_\bset)} \cup
    \inparen{\bigcup_{\bset \in \nonemptyleaves{\csetv}} (\bset^\structureS \cup B_\bset)}
    & \text{(Proposition \ref{prop:graph})}
    \\
    & = \bsetv^\structure \cup \csetv^\structure
    & \text{($\bsetv \in \vertices$, $\csetv \in \vertices$)}
  \end{align*}
\item $\asetv \opequal \bsetv \opinter \csetv$.
  We need to show $\asetv^\structure =  \bsetv^\structure \cap \csetv^\structure$.

  Let $\asetv \not\in \vertices$, $\bsetv \not\in \vertices$, and $\csetv \not\in \vertices$.
  Then,
  \begin{align*}
    \asetv^\structure
    & = \asetv^\structureS
    & \text{($\asetv \not\in \vertices$)}
    \\
    & = \bsetv^\structureS \cap \csetv^\structureS
    & \text{(Proposition \ref{prop:membermodel})}
    \\
    & = \bsetv^\structure \cap \csetv^\structure
    & \text{($\bsetv \not\in \vertices$, $\csetv \not\in \vertices$)}
  \end{align*}

  Let $\asetv \not\in \vertices$ and $\bsetv \not\in \vertices$, but $\csetv \in \vertices$.
  Then,
  \begin{align*}
    \bsetv^\structure \cap \csetv^\structure
    & = \bsetv^\structureS \cap \csetv^\structure
    & \text{($\bsetv \not\in \vertices$)}
    \\
    & = \bsetv^\structureS \cap \inparen{\csetv^\structureS \cup \bigcup_{\bset \in \nonemptyleaves{\csetv}} B_\bset}
    & \text{($\csetv \in \vertices$)}
    \\
    & = \bsetv^\structureS \cap \csetv^\structureS
    & \text{($\bsetv^\structureS \cap B_\bset = \emptyset$)}
    \\
    & = \asetv^\structureS
    & \text{(Proposition \ref{prop:membermodel})}
    \\
    & = \asetv^\structure
    & \text{($\asetv \not\in \vertices$)}
  \end{align*}
  If $\asetv \not\in \vertices$ and $\csetv \not\in \vertices$, but $\bsetv \in \vertices$; the reasoning is same as above.

  Otherwise, either $\asetv \in \vertices$ or both $\bsetv \in \vertices$ and $\csetv \in \vertices$.
  Then, from Rules \ruleIntroEqRight, \ruleIntroEqLeft, \ruleIntroEqInter\ and
  definition of $\add$, we know $\asetv$, $\bsetv$, \emph{and} $\csetv$ in $\vertices$.
  Then,
  \begin{align*}
    \bsetv^\structure \cap \csetv^\structure
    & =  \inparen{\bsetv^\structureS \cup \bigcup_{\bset \in \nonemptyleaves{\bsetv}} B_\bset}
    \cap \inparen{\csetv^\structureS \cup \bigcup_{\bset \in \nonemptyleaves{\csetv}} B_\bset}
    & \text{($\bsetv$, $\csetv$ in $\vertices$)}
    \\
    \intertext{As each $B_\bset$ is disjoint from all other sets, the above expression simplifies to:}
    & = \inparen{\bsetv^\structureS \cap \csetv^\structureS}
    \cup \bigcup_{\bset \in \nonemptyleaves{\bsetv} \cap \nonemptyleaves{\csetv}} B_\bset
    \\
    & = \asetv^\structureS
    \cup \bigcup_{\bset \in \nonemptyleaves{\asetv}} B_\bset
    & \text{(Propositions \ref{prop:membermodel} and \ref{prop:graph})}
    \\
    & = \asetv^\structure
    & \text{($\asetv \in \vertices$)}
  \end{align*}
\item $\asetv \opequal \bsetv \opsetminus \csetv$.
  We need to show $\asetv^\structure =  \bsetv^\structure \setminus \csetv^\structure$.

  Let $\asetv \not\in \vertices$, $\bsetv \not\in \vertices$, and $\csetv \not\in \vertices$.
  Then,
  \begin{align*}
    \asetv^\structure
    & = \asetv^\structureS
    & \text{($\asetv \not\in \vertices$)}
    \\
    & = \bsetv^\structureS \setminus \csetv^\structureS
    & \text{(Proposition \ref{prop:membermodel})}
    \\
    & = \bsetv^\structure \setminus \csetv^\structure
    & \text{($\bsetv \not\in \vertices$, $\csetv \not\in \vertices$)}
  \end{align*}

  Let $\asetv \not\in \vertices$ and $\bsetv \not\in \vertices$, but $\csetv \in \vertices$.
  Then,
  \begin{align*}
    \bsetv^\structure \setminus \csetv^\structure
    & = \bsetv^\structureS \setminus \csetv^\structure
    & \text{($\bsetv \not\in \vertices$)}
    \\
    & = \bsetv^\structureS \setminus \inparen{\csetv^\structureS \cup \bigcup_{\bset \in \nonemptyleaves{\csetv}} B_\bset}
    & \text{($\csetv \in \vertices$)}
    \\
    & = \bsetv^\structureS \setminus \csetv^\structureS
    & \text{($\bsetv^\structureS \setminus B_\bset = \bsetv^\structureS$)}
    \\
    & = \asetv^\structureS
    & \text{(Proposition \ref{prop:membermodel})}
    \\
    & = \asetv^\structure
    & \text{($\asetv \not\in \vertices$)}
  \end{align*}
  Note that in contrast to intersection, if $\asetv \not\in \vertices$, $\bsetv \in \vertices$, and $\csetv \not\in \vertices$,
  the above analysis does not apply. We do need to introduce and reason about the equality in the graph.

  Let $\asetv \in \vertices$ or $\bsetv \in \vertices$.
  From Rules \ruleIntroEqRight, \ruleIntroEqLeft, \ruleIntroEqSetDiff~and
  definition of $\add$
  we know $\asetv$, $\bsetv$, \emph{and} $\csetv$ in $\vertices$.
  Then,
  \begin{align*}
    \bsetv^\structure \setminus \csetv^\structure
    & =  \inparen{\bsetv^\structureS \cup \bigcup_{\bset \in \nonemptyleaves{\bsetv}} B_\bset}
    \setminus \inparen{\csetv^\structureS \cup \bigcup_{\bset \in \nonemptyleaves{\csetv}} B_\bset}
    & \text{($\bsetv$, $\csetv$ in $\vertices$)}
    \\
    \intertext{As each $B_\bset$ is disjoint from all other sets, the above expression simplifies to:}
    & = \inparen{\bsetv^\structureS \setminus \csetv^\structureS}
    \cup \bigcup_{\bset \in \nonemptyleaves{\bsetv} \setminus \nonemptyleaves{\csetv}} B_\bset
    \\
    & = \asetv^\structureS
    \cup \bigcup_{\bset \in \nonemptyleaves{\asetv}} B_\bset
    & \text{(Propositions \ref{prop:membermodel} and \ref{prop:graph})}
    \\
    & = \asetv^\structure
    & \text{($\asetv \not\in \vertices$)}
  \end{align*}

\item $x \opin S$, $x \opnotin S$.
  
  Note that irrespective of whether $\asetv \in \vertices$ or $\asetv \not\in \vertices$,
  $\asetv^\structureS \subseteq \asetv^\structure$. Thus, from Proposition \ref{prop:membermodel},
  $x^\structure \in \asetv^\structure$ if $x \opin \asetv$ is a constraint.

  It remains to show that if $x \opnotin S$ is a constraint then $x^\structure \not\in S^\structure$.
  If $\asetv \not\in \vertices$, then again $x^\structure \not\in \asetv^\structure$
  follows from Proposition \ref{prop:membermodel}.
  If $\asetv \in \vertices$, then observe that $\asetv^\structure$ is
  $\asetv^\structureS \cup \bigcup_{\bset \in \nonemptyleaves{\csetv}} B_\bset$.
  We already know $x^\structure \not\in \asetv^\structureS$.
  It remains to show that
  $x^\structure \not\in \bigcup_{\bset \in \nonemptyleaves{\csetv}} B_\bset$.
  This follows from the definition of $B_\bset$.
\item $c_S \opequal \opcard{S}$.

  From Proposition \ref{prop:graph}, we know that for $\bset, \cset$ in $\nonemptyleaves{\asetv}$:
  \[ \models_\settheory  \bset \opinter \cset \opequal \opemptyset \]
  and also,
  \[
  \models_\settheory \inparen{\bigwedge_{\bset \in E} \bset \opequal \opemptyset}
  \Rightarrow
  \inparen{\asetv \opequal \bigsqcup_{\bset \in \nonemptyleaves{\asetv}} \bset}
  \]
  where $E = \setbuilder{\bset \in \vertices}{\bset \opequal \opemptyset \in \setconstclosed}$.

  In $\structure$, as for each $\bset \in E$, $\bset^\structure = \emptyset$, it follows that:
  \[ \asetv^\structure = \bigcup_{\bset \in \nonemptyleaves{\asetv}} \bset^\structure\text{\quad.} \]
  Also, for $\bset, \cset$ in $\nonemptyleaves{\asetv}$:
  \[ \bset^\structure \cap \cset^\structure = \emptyset \text{\quad.}\]

  In other words, $\asetv^\structure$ is a disjoint union of $\bset^\structure$ where $\bset \in \nonemptyleaves{\asetv}$.
  It follows that,
  \[ \card{\asetv^\structure} = \sum_{\bset \in \nonemptyleaves{\asetv}} \card{\bset^\structure}\]
  For a leaf node $\bset \in \nonemptyleaves{\asetv}$, from (\ref{eq:settermcard}) we know that $\card{\bset^\structure} = \card{\bset^\structureS} + \card{B_\bset} = c_\bset^\structure$.
  We may thus conclude,
  \[  \card{\asetv^\structure} = \sum_{\bset \in \nonemptyleaves{\asetv}} c_\bset^\structure \]
  From the constraint on cardinality for $\asetv$ induced by the graph, i.e the constraint on $c_\asetv$ in $\grapharithconst$,
  we know that $c_\asetv^\structure  = \sum_{\bset \in \nonemptyleaves{\asetv}} c_\bset^\structure$.
  The result follows:
  \[ \card{\asetv^\structure} = c_\asetv^\structure \qedhere \]
\end{enumerate}

\end{proof}

%%% Local Variables:
%%% mode: latex
%%% TeX-master: "pap.tex"
%%% End:

%% \begin{prop} \label{prop:solution-soundness}
%% Let $\derivtree$ be a derivation tree with root  
%% $\langle \setconst_0,$ $\memconst_0,$ $\arithconst_0,$ $(\emptyset, \emptyset) \rangle$.
%% If $\derivtree$ has a branch saturated with respect to  rules $\mathcal{R}_1 \cup \mathcal{R}_2 \cup \mathcal{R}_3$,
%% then there exists a model $\structure$ of $\settheory$ 
%% that satisfies $\setconst_0 \cup \memconst_0 \cup \arithconst_0$.
%% \end{prop}%

%% \begin{proof}[Sketch]
%% We build the model of the leaf nodes in the graph by modifying as needed the model obtained 
%% from Proposition~\ref{prop:membermodel}.
%% We add additional elements to these sets to make the cardinalities match
%% the model satisfying the arithmetic constraints and the constraints induced by the graph.
%% Propositions \ref{prop:graph} and \ref{prop:arithmodel} ensure that 
%% it is always possible to do so without violating the set constraints.
%% \qed
%% \end{proof}

\begin{prop}[Completeness]
\label{prop:completeness}
Under any fair derivation strategy,
every derivation of a set $\mathcal C$ of $\settheory$-unsatisfiable constraints 
extends to a refutation. 
\end{prop}

\begin{proof}
Contrapositively, suppose that $\mathcal C$ has a derivation $\derivation$  
that cannot be extended to a refutation.
By Proposition~\ref{prop:finite-derivations}, $\derivation$ must be extensible to one
that ends with a tree with a saturated branch.
By Proposition~\ref{prop:solution-soundness}, $\mathcal C$ is satisfiable in $\settheory$.
%\qed
\end{proof}

%------------------------------------------------------------------------------
\subsection{Soundness}
%------------------------------------------------------------------------------

We start by showing that every rule preserves constraint satisfiability. 

\begin{lem} \label{lem:soundness}
For every rule of the calculus, the premise state is satisfied 
by a model $\structure_p$ of $\settheory$
iff
one of its conclusion configurations is satisfied 
by a model $\structure_c$ of $\settheory$
where $\structure_p$ and $\structure_c$ agree on the variables shared by the two states.
\end{lem}

\begin{proof}[Sketch]
%\ctsaysi{This needs revising}
Soundness of the rules in Figure~\ref{fig:set-rules1} and Figure~\ref{fig:set-rules2}
follows trivially from the semantics of set operators and the definition of $\setconstclosed$.
%Soundness of \ruleSetDiseq\ also follows easily from a case analysis.
%Likewise for the optional propagation rule \refrule{rule:singandunion}.
%
Soundness of \ruleMergeEqI\ 
follows from properties of the graph (see Proposition \ref{prop:graph}, in particular the property
that leaf terms are disjoint).
The rules in Figure~\ref{fig:introduce} and rule \ruleMergeEqII\ 
do not
modify the constraints, but we need them to establish properties of the graph.
Soundness of the induced graph constraints in \ruleAriContra\ follows from
Proposition \ref{prop:graph} (in particular properties
\ref{prop:graph:propdisjoint} and \ref{prop:graph:propleaves}).
Soundness of \rulePropMinsize\ follows from the semantics of cardinality.
Soundness of \ruleGuessEmpty, \ruleMemArrange\ and \ruleGuessLB\ 
is trivial.
%\qed
\end{proof}

\begin{prop}[Soundness]
\label{prop:soundness}
Every set of $\settheory$-constraints that has a refutation is $\settheory$-unsatisfiable.
\end{prop}

\begin{proof}[Sketch]
Given Lemma~\ref{lem:soundness},
one can show by structural induction on derivation trees 
that the root of any closed derivation tree is $\settheory$-unsatisfiable. 
The claim then follows from the fact that every refutation of a set $\mathcal C$ of 
$\settheory$-constraints starts with a state $\settheory$-equisatisfiable with $\mathcal C$.
%\qed
\end{proof}

%%% Local Variables:
%%% mode: latex
%%% TeX-master: "pap.tex"
%%% End:

%!TEX root =  pap.tex

%\ifTHESIS
%\section{Theory of finite sets with cardinality}
%We have implemented a decision procedure based on the rules in Chapter
%\ref{chap:sets} in the SMT solver CVC4.
%\else
\section{Evaluation}
We have implemented a decision procedure based on the calculus above
in the SMT solver \cvc~\cite{CVC4-CAV-11}. We describe a high-level, 
non-deterministic version of it here, followed by an experimental evaluation
on benchmarks from program analysis.

\subsection{Derivation strategy}
\reviewerbis{
Regarding the strategy: it's interesting that branching is preferred to merging...

Also, I was wondering if in (2), the application of propagation rules for $R_1$,
you could not have the special instance of SET DISEQUALITY for disequalities of
the form ($t\not\sim\emptyset$), which is not really a branching rule, given that
one branch is immediately going to be closed. Such instances are presumably
going to be numerous, given the GUESS EMPTY SET rule.
}
The decision procedure can be thought of as a specific strategy for
applying the rules given in Section~\ref{sec:sets:calculus},
divided into the sets $\mathcal{R}_1$, \ldots, $\mathcal{R}_4$
introduced in Section~\ref{sec:sets:correct}.

Our derivation strategy can be summarized as follows. We start
the derivation from the initial state $\langle \setconst_0,
\memconst_0, \arithconst_0, \graphconst_0 \rangle$ with
$\graphconst_0$ the empty graph, as described in
Section~\ref{sec:sets:calculus}, and apply the steps listed below, in
the given order.
The steps are described as rules being applied to a \emph{current}
branch of the derivation tree being constructed. 
Initially, the current branch is the only branch in the tree. On
application of a rule with more than one conclusion, we select one of
the branches (say, the left branch) as the current branch.
\begin{enumerate}
\item If a rule that derives \textsf{unsat} is applicable to the
  current branch, we apply one and close the branch. We then pick
  another open branch as the current branch and repeat Step 1. If no
  open branch exists, we stop and output \textsf{unsat}.
\item If a \emph{propagation} rule (those with one conclusion) in
  $\mathcal{R}_1$ is applicable, apply one and go to Step 1.
\item If a \emph{split} rule (those with more than one conclusion) in
  $\mathcal{R}_1$ is applicable, apply one and go to Step 1.
\item If \ruleGuessEmpty~rule is applicable, apply it and go to Step
  1.
\item If an introduce or merge rule in $\mathcal{R}_2$ is applicable,
  apply it and go to Step 1.
\item If any of the remaining rules are applicable, apply one and go to
  Step 1.
\item
  At this point, the current branch is saturated. Stop and output \textsf{sat}.
\end{enumerate}
Note that 
%%CT not sure why this is important here
%rules in $\mathcal{R}_1$ do not introduce any new set terms. 
%Furthermore, 
if there are no constraints involving the cardinality operator, 
then steps 1 to 3 above are sufficient for completeness.

\begin{table}[t]
\caption{Performance of our calculus on benchmarks derived from verification of programs}
\label{table:sets:card}
\center
\begin{tabular}{|l|l|r|r|r|}
\hline
{\bf file} & {\bf output} & {\bf time (s.)} & {\bf \# vertices} & {\bf \# leaves} \\
\hline
\hline
cade07-vc1.smt2           & unsat           & 0.00   & 3          & 3 \\
cade07-vc2a.smt2          & unsat           & 0.00   & 6          & 3 \\
cade07-vc2b.smt2          & sat             & 0.01   & 15         & 5 \\
cade07-vc2.smt2           & unsat           & 0.01   & 6          & 3 \\
cade07-vc3a.smt2          & unsat           & 0.00   & 6          & 0 \\
cade07-vc3b.smt2          & sat             & 0.02   & 15         & 6 \\
cade07-vc3.smt2           & unsat           & 0.01   & 6          & 0 \\
cade07-vc4b.smt2          & sat             & 0.16   & 44         & 12 \\
cade07-vc4.smt2           & unsat           & 0.17   & 51         & 16 \\
cade07-vc5b.smt2          & sat             & 0.39   & 63         & 21 \\
cade07-vc5.smt2           & unsat           & 0.38   & 77         & 25 \\
cade07-vc6a.smt2          & unsat           & 0.02   & 32         & 12 \\
cade07-vc6b.smt2          & sat             & 0.04   & 32         & 12 \\
cade07-vc6c.smt2          & sat             & 0.06   & 32         & 12 \\
cade07-vc6.smt2           & unsat           & 0.32   & 36         & 16 \\
\hline
\hline
cvc4-card.scala-10.smt2   & 2 sat/2 unsat   & 0.10   & 48         & 19 \\
cvc4-card.scala-12.smt2   & 1 sat/3 unsat   & 0.03   & 0          & 0 \\
cvc4-card.scala-14.smt2   & 2 sat/2 unsat   & 0.09   & 25         & 11 \\
cvc4-card.scala-15.smt2   & 1 sat/3 unsat   & 0.01   & 0          & 0 \\
cvc4-card.scala-16.smt2   & 2 sat/4 unsat   & 0.26   & 39         & 18 \\
cvc4-card.scala-17.smt2   & 1 sat/3 unsat   & 0.02   & 19         & 8 \\
cvc4-card.scala-18.smt2   & 2 sat/2 unsat   & 0.10   & 39         & 20 \\
cvc4-card.scala-21.smt2   & 2 sat/2 unsat   & 1.69   & 134        & 35 \\
cvc4-card.scala-6.smt2    & 1 sat/4 unsat   & 0.02   & 8          & 5 \\
cvc4-card.scala-8.smt2    & 1 sat/3 unsat   & 0.06   & 21         & 12 \\
\hline
\end{tabular}

\end{table}
% \reviewer{
%  - page 28: "Figure 10" is in fact table and not a figure.
% }

%%% Local Variables:
%%% mode: latex
%%% TeX-master: "pap.tex"
%%% End:

\subsection{Experimental evaluation}
%% \ctsaysi{A few words on the execution environment.}
We evaluated our procedure on benchmarks obtained from a software verification 
applications. The experiments were run on a machine with 3.40GHz Intel i7
CPU with a memory limit of 3 GB and timeout of 300 seconds. We used a
development version of \cvc for this evaluation.\footnote{%
\url{https://github.com/kbansal/CVC4/tree/37f6117}}
Benchmarks are available on the \cvc website.\footnote{\url{http://cvc4.cs.stanford.edu/papers/LMCS-2018/}}

The first set of benchmarks consists of single query benchmarks obtained from
verifying programs manipulating pointer-based data structures. These
were generated by the Jahob system, and have been used to evaluate earlier work on
decision procedures for finite sets and cardinality~\cite{KNR06,KR07,SSK11}.
The results from running \cvc on these benchmarks are provided in the top half of Table~\ref{table:sets:card}. 
The output reported by \cvc is in the second column. The third column shows the solving time. 
The fourth and fifth columns
give the maximum number of vertices (\# V) and leaves\footnote{The \# L
  statistic is updated only when explicitly computed, so the numbers
  are approximate. For the same reason, \# L is 0 on
  certain benchmarks even though \# V is not.
  This is because \cvc was able to report unsat before the need for computing 
  the set of leaves arose.}  (\# L)
in the graph at any point during the run of the algorithm. 
Keeping the number of leaves low is important to avoid a blowup from the \ruleMergeEqII~rule.

Although we have not rerun the systems described in~\cite{KNR06,KR07,SSK11},
we report here the experimental results as stated in the respective
papers.\footnote{One reason we were unable to do a more thorough
  comparison with previous work is that those implementations are no longer being maintained.}
Since the experiments were run on different machines the
comparison is only indicative, but it does suggest that our solver has
comparable performance.

%% In \cite{KR07}, the
%% running time of the algorithm in~\cite{KNR06} for \textsf{BAPA} fragment is compared with algorithm in
%% \cite{KR07}.
In \cite{KR07}, the procedure from \cite{KNR06} is reported to solve 12 of the 15 benchmarks with a timeout of 100 seconds,
while the novel procedure in \cite{KR07} is reported to solve 11 of the 15 benchmarks with the same timeout.
The best-performing previous procedure (\cite{SSK11}) can solve all 15 benchmarks in under a second.\footnote{
%
%% An interesting statistic reported in \cite{SSK11} is the number of
%% Venn regions the algorithm considered. For example, on vc6a, vc6b and
%% vc6c benchmarks, the algorithm considers over 500 regions.
%
\cite{SSK11} includes a second set of benchmarks, but we were unable
to evaluate our procedure on these, as they were only made available in a non-standard
format and were missing crucial datatype declarations.}
As another point of comparison, we tested the procedure from~\cite{SSK11} on a benchmark of the
type mentioned in Section 1.1: a single
constraint of the form $x \opin A_1 \opunion \ldots \opunion A_{21}$.
As expected, the solver failed (it ran out of memory after 85 seconds). In
contrast, \cvc solves this problem instantaneously. % as expected.

Finally, another important difference compared to earlier work is that our
implementation is completely integrated in an actively developed and maintained
solver, \cvc.

To highlight the usefulness of an implementation in a full-featured SMT solver,
we did a second evaluation on a set of incremental (i.e., multiple-query) benchmarks 
obtained from the Leon verification system~\cite{BKKS13}. 
These contain a mix of membership and cardinality constraints combined with
constraints over the theories of datatypes and bitvectors.
The results of this evaluation are shown in the bottom half of Table~\ref{table:sets:card}.  The output column reports the number of sat and unsat queries in each benchmark.
\cvc successfully solves all of the queries in these benchmarks in under one second.  To the best
of our knowledge, no other SMT solver can handle this combination of theories.

\section{Conclusion}
We presented a new decision procedure for deciding
finite sets with cardinality constraints and proved its correctness.
A novel feature of the procedure is that it can reason directly and efficiently
about both membership constraints and cardinality constraints.
We have implemented the procedure in the SMT solver \cvc, and
demonstrated the feasibility as well as some advantages of our approach.
We hope this work will enable the use of sets and cardinality constraints in many new
applications that rely on SMT solvers.
We also expect to use it to drive the development of a standard
theory of sets under the SMT-LIB initiative~\cite{BarFT-SMTLIB}.

We expect to pursue several directions of future work.  We will investigate
relaxing Restriction 3.1 by doing more reasoning modulo equality.  We will also
experiment with different strategies to attempt to find the most efficient
ones.  We will also look into efficient means of combining sets with other
theories and investigate extensions to relations and relational operators.

%%% Local Variables:
%%% mode: latex
%%% TeX-master: "pap.tex"
%%% End:

\section*{Acknowledgement}
   The authors wish to acknowledge fruitful discussions with Viktor
   Kuncak and Etienne Kneuss and for providing the Leon benchmarks. We
   thank Philippe Suter for his help running the algorithm
   from~\cite{SSK11}.

%% in general the use of bibtex is encouraged

\bibliographystyle{alpha}
\bibliography{biblio}
  
%% \begin{thebibliography}{Kos97}

%% \bibitem[Kos97]{koslowski:mib}
%% J{\&quot;u}rgen Koslowski.
%% \newblock Monads and interpolads in bicategories.
%% \newblock {\em Theory Appl. Categ.}, 3(8):182--212, 1997.

%% \end{thebibliography}

% K: Have incorporated everything from appendix in main files, but
%have left the file around, because some of the statements are more
%verbose than the ones in main paper. Just in case we need it.
%\appendix \input{appendix}

%% \section{}
%%   Doctors recommend to take out the appendix if pain starts to get
%%   life threatening.

\end{document}